\documentclass[11pt]{article}
\usepackage[OT1]{fontenc}
\usepackage{dsfont}
\usepackage{fullpage}
\usepackage{booktabs} 
\usepackage[ruled]{algorithm2e} 

\SetAlFnt{\small}
\SetAlCapFnt{\small}
\SetAlCapNameFnt{\small}
\SetAlCapHSkip{0pt}
\IncMargin{-\parindent}

\usepackage{amssymb,amsthm,amsmath}
\usepackage{tikz}
\usepackage{bbm}
\usepackage{bm}
\usepackage{dsfont}
\usepackage{hyperref}
\usepackage[square,sort,comma,numbers]{natbib}
\usepackage{url}            
\usepackage{microtype}      
\usepackage{xcolor}
\usepackage{epic}
\usepackage{epsfig}
\usepackage{verbatim}
\usepackage[justification=centering]{caption}
\usepackage{enumitem}
\usepackage[normalem]{ulem}
\usepackage{thmtools,thm-restate}
\usepackage{float}
\usepackage{multirow}
\usepackage{makecell}
\usepackage{subcaption}
\usepackage{endnotes}

\input{style.sty}
\input{macros.sty}

\newcommand{\BSE}{\texttt{BSE}}

\newcommand{\DSE}{\texttt{DSE}}
\newcommand{\SSE}{\texttt{SSE}}
\newcommand{\RME}{\texttt{RME}}

\newcommand{\MILP}{\texttt{MILP}}
\newcommand{\MIP}{\texttt{MIP}}

\newcommand{\Markov}{\textsc{Markovian}}
\newcommand{\First}{\textsc{First}-$k$}

\title{Is Learning Effective in Dynamic Strategic Interactions? Evidence from Stackelberg Games\thanks{This work is supported by  Army Research Office
Award W911NF-23-1-0030, ONR Award N00014-23-1-2802 and NSF Award CCF-2303372.}}
\author{
Michael Albert \\ University of Virginia \\ \small albertm@darden.virginia.edu
\and
Quinlan Dawkins\thanks{Work   done when the author is at UVA CS.} \\ Advanced Micro Devices \\ \small qdawkins@amd.com
\and
Minbiao Han \\ University of Chicago \\ \small minbiaohan@uchicago.edu
\and
Haifeng Xu \\ University of Chicago \\ \small haifengxu@uchicago.edu
}

\date{}

\begin{document}


\maketitle
\begin{abstract}
In many settings of interest, a policy is set by one party, the leader, in order to influence the action of another party, the follower, where the follower's response is determined by some private information.
A natural question to ask is, can the leader improve their strategy by learning about the unknown follower through repeated interactions?
A well known folk theorem from dynamic pricing, a special case of this leader-follower setting, would suggest that the leader cannot learn effectively from the follower when the follower is fully strategic, leading to a large literature on learning in strategic settings that relies on limiting the strategic space of the follower in order to provide positive results.
In this paper, we study dynamic Bayesian Stackelberg games, where a leader and a \emph{fully strategic} follower interact repeatedly, with the follower's type unknown. 
Contrary to existing results, we show that the leader can improve their utility through learning in repeated play. 
Using a novel average-case analysis, we demonstrate that learning is effective in these settings, without needing to weaken the follower's strategic space. 
Importantly, this improvement is not solely due to the leader's ability to commit, nor does learning simply substitute for communication between the parties.
We provide an algorithm, based on a mixed-integer linear program, to compute the optimal leader policy in these games and develop heuristic algorithms to approximate the optimal dynamic policy more efficiently. 
Through simulations, we compare the efficiency and runtime of these algorithms against static policies.
\end{abstract}


\section{Introduction}
\label{sec:introduction}

Optimal pricing given unknown demand is a well-studied problem that comes up in many settings, such as airline ticket pricing, ride-sharing platforms, and online retail, where pricing strategies are crucial for maximizing revenue and efficiency \citep{vickrey1961counterspeculation,aggarwal2006truthful,varian2009online}. 
A natural extension of the problem is to assume that there is an opportunity for a seller to learn the demand through repeated interactions using \emph{dynamic pricing}.
However, even in the simplest setting, a single buyer with a fixed valuation drawn from a known distribution, i.e. binary demand, it is impossible for a seller to exploit any information learned to improve her revenue.
Formally, no dynamic pricing policy can outperform the fixed (or \emph{static}) policy of simply offering the optimal single round price, i.e. the \emph{Myerson price} \citep{myerson1981optimal}, at every round, at least when the buyer is both strategic and patient.\endnote{This result is a well known \textit{folk theorem}, and therefore lacks an agreed upon citation for the initial statement of the theorem. However, the reader can consult \citet{pavan2014dynamic,devanur2014perfect, vanunts2019optimal} for recent relevant work.}

This negative result, which we refer to as the \emph{No Learning Theorem}, is a function of the buyer behaving strategically in his purchase decisions in order to influence future prices.
The effect on the literature of this strong negative result has been to adopt assumptions that weaken the strategic behavior of the buyer (e.g., myopic buyers) in order to obtain positive results \citep{amin2013learning, amin2014repeated,dawkins2022first,vanunts2019optimal,immorlica2017repeated,dawkins2021limits}.
Given this foundational negative result, a natural question arises, ``Does a static policy suffice to achieve optimality
in a dynamic setting with strategic agents?''
Stated another way, ``Generally, can one both learn and exploit the information learned when facing strategic agents if we move beyond the dynamic pricing problem?''

There are many settings outside of dynamic pricing that share similar strategic concerns, specifically one party, the leader, sets a policy that induces a response by another party, the follower.
However, the follower may have some private information that would inform the optimal leaders policy while also having an incentive to mislead the leader due to conflicting objectives.
For example, in contract design \citep{bolton2004contract} the principal, the leader, is designing a contract to induce certain actions in the agent, the follower.
The contract may be revisited periodically, allowing for opportunities for the principal to redesign the contract given the revealed information from previous contracting periods.
In the context of security, whether it be cyber security, border patrol, airport security, protecting wildlife from poachers, or many other security contexts, the defender, the leader, commits to a security policy and the attacker, the follower, responds \citep{paruchuri2008playing,yang2014adaptive}.
However, the value of different targets to the attacker may be unknown, but through repeated rounds of defense and attack, the value may be learned based on the actions of the attacker.
Tolling, a pricing based congestion management approach, is also consistent with this setting.
The central planner, the leader, sets tolls along road segments under which the traffic flow, the follower, best responds to both minimize latency plus cost between any two nodes \citep{roth2016watch}.
Here again, the central planner may learn about the distribution of traffic over time through observing latency throughout the network.

The preceding examples, including the dynamic pricing problem, can all be formalized using the framework of Bayesian Stackelberg games \citep{conitzer2006computing}. 
Specifically, a standard Stackelberg game models a two-step sequential decision-making process between two agents, a leader and a follower.
In the dynamic pricing problem, the leader would be the seller and the buyer would be the follower.
The leader moves first in a Stackelberg game by committing to a randomized strategy, e.g. randomizing over prices for an item.
Then the follower responds with a utility maximizing action, e.g. does or does not buy the item, after observing the leader's strategy. 
However, the follower's utility information, the valuation of the buyer in the dynamic pricing problem, may be private and unknown to the leader, and instead the leader may only know a distribution over possible follower types, hence the \emph{Bayesian} Stackelberg game.
In repeated settings there is an opportunity to learn from the behavior of the follower.
In this general framework, given the strong strategic position of the follower, can the leader effectively exploit this information, in contrast to the No Learning Theorem for the dynamic pricing problem?


\subsection{Contributions}
\label{subsec:contributions}



This paper has three main contributions. Our first contribution is to show that learning and exploiting the learned information—referred to as \emph{effective learning}—is likely to be achievable, particularly in games with many possible follower types and ``random'' payoff matrices.
We demonstrate this through an \emph{average case} analysis for the space of dynamic Bayesian Stackelberg games.
Specifically, we define the concept of a random game and impose a condition on the distribution generating this random game.
We then prove using a novel stochastic geometry based argument that the probability of the leader being able to learn effectively converges to one linearly with the number of follower types.
This average case analysis highlights that the dynamic pricing problem is an anomaly within the broader set of dynamic Bayesian Stackelberg games, suggesting that the insights drawn from the No Learning Theorem do not generally apply in symmetric strategic settings.
In the process of proving this claim, we develop a sufficient condition that guarantees effective learning by the leader, and we show that this condition holds with high probability.
Despite the condition being merely sufficient, we consider it an important contribution for characterizing the types of games where learning is effective.

While we demonstrate that learning is generally effective in arbitrary repeated Bayesian Stackelberg games, our setting differs from the broader mechanism design literature in that a Bayesian Stackelberg game does not allow for arbitrary communication between the leader and follower.
One potential interpretation of our first result might be that effective learning is simply a substitute for direct communication.
If so, and if the follower could directly communicate their type, the leader might achieve a strict improvement over the optimal dynamic policy.
As our second contribution, we show that this is not the case.
In fact, we prove that the optimal dynamic policy is at least as effective as direct communication under a static policy, and we demonstrate scenarios in which the leader can learn a policy that achieves an $\Omega(1)$ improvement over the static policy.

Finally, while our previously stated contributions use constructive techniques to demonstrate the positive results, they do not generally provide the optimal dynamic policy.
As our final contribution, we develop a mixed-integer linear program to compute the optimal dynamic policy.
Unfortunately, it is straightforward to show that finding the optimal policy is NP-hard in general.
Therefore, we also develop two heuristics: one based on a Markov approach, and another that uses the first $k$ rounds of the game to learn, followed by an exploitation phase.
We show that these heuristics perform well on a set of standard Stackelberg games adapted for the repeated Bayesian setting.
Both heuristics approximate the optimal dynamic leader policy effectively and are significantly more computationally efficient, with the First-$k$ rounds policy consistently outperforming the Markovian policy in our experiments, while also being a polynomial-time algorithm for fixed $k$.

Together, these contributions suggest that in fully strategic settings beyond the pricing problem, learning is likely both effective and feasible.
Thus, our results imply that in fully strategic settings outside of dynamic pricing, it is unnecessary to weaken the strategic assumptions on the follower to achieve positive results regarding learning, contrary to the standard approach in the literature.

\subsection{Related Work}
\label{subsec:related_work}

In this section, we discuss the relationship between this work and several distinct strands in the literature.
Given the foundational nature of learning in strategic settings, particularly around pricing, this work connects to a large body of existing work.

While our setting is more restrictive than general mechanism design due to the limited ability for the leader and follower to communicate, it is closely related to the literature on dynamic mechanism design.
The most closely related work is by \citet{balseiroFutilityDynamicsRobust2021a}, who study the general dynamic mechanism design problem with a single principal (the leader) and a single agent (the follower).
The agent receives shocks to their utility function over time, drawn from an unknown distribution.
The principal aims to maximize their \emph{worst-case} utility for all distributions within a specified set.
In contrast, our setting focuses on maximizing the principal's average utility given a prior distribution over agent types.
Interestingly, our results differ from those of \citet{balseiroFutilityDynamicsRobust2021a}.
Specifically, they show that, under two assumptions, the optimal static mechanism achieves the minimax utility for the principal, whereas we demonstrate that dynamic policies increase average utility for most dynamic Bayesian Stackelberg games.
Our setting satisfies their two assumptions, suggesting that the fundamental difference lies in the maximin versus expected utility formulations.
This implies that, in settings with a prior over agent types, learning is likely effective, while in settings with a maximin objective—often used for robustness—learning is generally not effective.

Our setting, if viewed as a restricted dynamic mechanism design problem, is also related to \citet{pavan2014dynamic}.
The No Learning Theorem can be seen as a direct consequence of the characterization of the optimal mechanism given by \citet{pavan2014dynamic}, although their work does not address the likelihood that learning is generally effective.
Similar results to the No Learning Theorem have been independently identified in the literature, going back to at least \citet{baron1984regulation}.
\citet{baron1984regulation} studied a two-period interaction between a regulator and a regulated firm, where the agent’s type in the second round evolves from the first round.
They show that the multi-period problem reduces to a static problem when the agent type remains the same across periods, consistent with the insights from the No Learning Theorem.
Other work, such as \citet{courty2000sequential}, has studied the optimal mechanism for a two-period ticket-selling problem.
In contrast to our work, their setting involves an agent who only knows the distribution of their valuations in the first period and learns their actual valuations in the second period.
In our setting, the agent knows their type from the beginning, and it remains constant throughout all rounds.
\citet{battaglini2005long} examined the problem of characterizing the optimal contract between a firm and a long-term customer.
They show that when the customer's type is not constant, the optimal contract is not static.
In our setting, we look at the case where the follower's type is constant, and the optimal policy is, in general, not static.
Overall, our work primarily differs in that we focus on the prevalence of learning in strategic settings, while the previously mentioned works either broadly characterize optimal mechanisms \citep{pavanDynamicMechanismDesign2014a} or investigate the possibility of learning in specific instances \citep{baron1984regulation,courty2000sequential,battaglini2005long}.

In this work, we assume that the leader can fully commit to their strategy, while several works in the dynamic mechanism design literature consider cases where the leader cannot commit.
When the leader cannot commit, it leads to the well-known "ratchet effect" (see, for example, \citet{freixas1985planning,laffont1988dynamics}).
In this scenario, the ability to learn is diminished due to the inability to credibly commit to respond appropriately to any revealed information.
This is precisely what \citet{freixas1985planning,laffont1988dynamics} identify in a two-stage interaction between a principal and an agent, where the agent is incentivized to underproduce to avoid a demanding production schedule in future rounds.
We are interested in the possibility of learning in settings with full strategic power, so we consider the ability to commit as a reasonable assumption.
This assumption is also widely used in the literature on dynamic principal-agent problems in both mechanism design \citep{kakade2013optimal,pavan2014dynamic} and Stackelberg games \citep{li2017review,lauffer2022no}.

\subsubsection{Learning in Stackelberg Games}
\label{subsubsec:learning_in_stackelberg_games}

In addition to the broader literature on dynamic mechanism design, there has been significant interest in learning in the narrower space of Stackelberg games.
When the follower's payoff information is unknown, \citet{letchford2009learning,balcan2014learning,peng2019learning,haghtalab2024calibrated,han2024learning} propose various policies that allow a leader to learn the follower's utility by observing the follower's response to a particular leader's strategy.
Notably, these policies can learn the optimal leader strategy efficiently with a polynomial number of learning rounds, but
these papers all assume that the follower behaves myopically, i.e. the follower always best responds to the posted leader strategy without regard for future rounds.
Therefore, they are not characterizing the possiblity of learning in strategic settings.

Recently, \citet{haghtalab2022learning} studied learning in Stackelberg games with non-myopic agents and proposed a no-regret learning policy for the leader.
However, their work relies on the assumption that the follower discounts the future utility at a greater rate than the leader, a weakening of the strategic assumptions on the follower.
In our work, we assume symmetric discounting.
In a slightly different direction, \citet{deng2019strategizing} studies dynamic policy design in a repeated setting when the follower's payoff information is public but instead of best responding the follower follows a no-regret learning algorithm to interact with the leader.
While this can be construed as a certain kind of limited strategic behavior, it is very different than the fully strategic setting we are studying.

In \cite{ananthakrishnanKnowledgePowerIm2024}, they consider the case of a ``meta game'' between two players who play an unknown, but drawn from a known distribution, Stackelberg game.
They examine cases under which the players, who may be differentially informed, deploy a multi-round strategy as an action in the meta-game.
They show, consistent with the No Learning Theorem, that strategic interactions between these uninformed players cannot in all cases learn, and exploit the information learned, about the game in order to achieve their optimal Stackelberg strategy.
In contrast, we conduct an average case analysis where we show that it is indeed possible for an uninformed leader to learn effectively against a fully informed follower in the majority of games.

\subsubsection{Learning Optimal Prices}
\label{subsubsec:learning_optimal_prices}

The motivation for this work was the No Learning Theorem in dynamic pricing, and therefore the literature on learning optimal prices is closely related.
Given the breadth of this literature we will be restricting our discussion to work on learning in strategic settings.
Given that the No Learning Theorem places a strict impossibility for the fully strategic setting, work has generally been focused on the relaxation of some aspect of the strategic capacity of the buyer.
When the buyer discounts the future at a greater rate than the seller, there are positive results \citep{amin2013learning, amin2014repeated,mohri2014optimal,mohri2015revenue,golrezaeiDynamicIncentiveAwareLearning2021a}, e.g. there exists a no-regret learning policy to learn the optimal price.
Other work has focused on relaxing the commitment assumption in pricing \citep{immorlica2017repeated} or restricting the strategy space in some way \citep{dawkins2021limits,dawkins2022first}.
In settings with restricted commitment or strategy spaces, learning is generally possible.
However, in contrast to this literature we are interested in settings beyond the standard pricing problem, and we demonstrate that even in settings with full strategic power and symmetric discount rates, deviations from the classic pricing problem lead to the ability to learn effectively.

\subsubsection{Other Related Work}
\label{subsubsec:other_related_work}

A more recent literature has developed around algorithmic contract design problems. 
Most of this literature has focused on the computational issues of contract design \citep{dutting2019simple,alon2021contracts,castiglioni2022designing,castiglioni2021bayesian}. 
In recent work, the sample complexity of online contract design has been explored by \cite{ho2014adaptive,zhu2022sample}.
\cite{cohen2022learning} further extend the analysis to a specific scenario where the agent's utility function exhibits bounded risk aversion.
However, this learning does not happen in a strategic setting, which is our main focus.


Recent work has also considered the computation \citep{vorobeychik2012computing,bensoussan2015maximum,chang2015leader,goktas2022zero,vu2022stackelberg} of Stackelberg equilibria in stochastic games (see \cite{shapley1953stochastic} for a definition of stochastic games).  
However, in this literature, there is generally no direct information asymmetry.
While play proceeds stochastically, the leader is not directly learning from the actions of the follower, distinct from our setting where the information asymmetry is core to the problem.

Finally, our setting is reminiscent of a traditional bandit problem \citep{slivkins2019introduction}. 
However, the key difference between the traditional bandit problem and our setting is that the arm does not behave strategically. 
The strategic consideration is core to our focus.
Therefore, traditional no-regret learning results do not apply in our setting, and in fact, the No Learning Theorem states that for at least some settings (such as the dynamic pricing problem) the regret is unbounded.

\section{Preliminaries and Problem Setup}
\label{sec:preliminaries_and_problem_setup}

In this section, we introduce notation and we formally define the notion of a dynamic Bayesian Stackelberg game along with the optimal solution concept.

\subsection{Stackelberg Games}
\label{subsed:stackelberg_games}
We consider a Stackelberg game played by two players, who are referred to as the \textit{leader (she)} and the \textit{follower (he)}. Each player has a finite action set, and we denote the leader's action set as $[m] = \{1,\cdots m\}$ and the follower's action set as $[n] = \{1,\cdots, n\}$. The leader's utilities are described by a matrix $R \in \mathbb{R}^{m \times n}$, in which $R_{i,j}$ is the leader utility when she plays action $i \in [m]$ and the follower responds with action $j \in [n]$. Similarly, the follower's utility matrix is denoted by $C \in\mathbb{R}^{m \times n}$.
We denote the Stackelberg game as $\{R, C\}$.
Without loss of generality, we assume $R_{i,j}, C_{i,j} \in [0,1], \text{ for all } i, j$ in all statements of theorems.
However, for example, we may allow $R_{i,j}, C_{i,j} \in \mathbb{R}$ for the sake of expositional clarity.

In a single round (i.e. static) Stackelberg game \citep{stackelberg1934marktform}, the leader moves first by committing to a \emph{mixed strategy} $\bx \in \Delta^m$, where $\Delta^{m} = \{\bvec{x}: \sum_{i\in[m]} x_i = 1 \text{ and } 0\leq x_i \leq 1\}$ is the $m$-dimensional simplex and each $x_i$ denotes the probability the leader plays action $i$. 
After observing the leader strategy $\bx$, the follower responds by playing some action   $j$ which leads to expected follower utility   $V(\bx, j) = \sum_{i \in [m]}   C_{i,j} x_i$. A rational follower will pick the optimal action $j^*(\bx) = \arg\max_{j  \in [n]} V (\bx, j)$ to maximize his own utility.\endnote{Since the follower moves after the leader, there is no need for the follower to randomize his strategy. If there are ties, we assume the follower breaks the tie in favor of the leader, which is without loss of generality \citep{von2004leadership}.}  
The leader's utility $U(\bx, j)$ is defined similarly.  When the follower's utility matrix $C$ is known to the leader, the leader can predict the follower's reaction $j^*(\bx)$ to any $\bx$. 
The rational leader will commit to $\bx^* = \arg\max_{\bx \in \Delta^m} U(\bx, j^*(\bx))$. 
This bi-level optimization problem can be solved in $\poly(m,n)$  time via linear programming \citep{conitzer2006computing}, and the optimal $\bx^*$ is the \emph{Strong Stackelberg Equilibrium} ($\SSE$). 

\subsubsection{Bayesian Stackelberg Games}
\label{subsubsed:bayesian_stackelberg_games}

In many settings of interest, the leader may not know the follower's payoff matrix $C$. 
To capture the leader's uncertainty about follower payoffs, we follow the literature and model the uncertainty as a random \emph{follower type} $\theta \in \Theta$ which is known privately to the follower while the leader only has a prior distribution $\bmu \in \Delta^{|\Theta|}$ over the types. 
This prior distribution, $\boldsymbol{\mu}$, may be a result of either leader beliefs or may be due to observing a population of followers from which the prior can be learned.
We denote the payoff matrix of a $\theta$-type follower as $C^{\theta}$, and the best response for follower type $\theta$ to any leader strategy $\bx$  as  $j^{*\theta}(\bx) = \arg\max_{j  \in [n]} V^{\theta} (\bx, j)$. 
As a natural extension of $\SSE$ to this Bayesian setup, a leader with prior knowledge $\bmu$ will play an $\bx^*$ to maximize her expected utility, formally,  $\bx^* = \arg\max_{\bx \in \Delta^m} \sum_{\theta \in \Theta } \mu(\theta) U^{\theta} (\bx, j^{*\theta}(\bx))$, known as the Bayesian Stackelberg Equilibrium ($\BSE$). 
Unlike the $\SSE$, computing a $\BSE$ in Bayesian Stackelberg games is NP-hard \citep{conitzer2006computing}.  
For notational convenience,  we drop the leader utility's dependence on the follower type $\theta$ and simply use $R, U(\bx)$  instead of $R^{\theta}, U^{\theta}(\bx)$. All results generalize trivially.
We denote a specific Bayesian Stackelberg game as $\{R, \Theta, \{C^\theta\}_{\theta \in \Theta}, \bmu\}$.

\subsubsection{Dynamic Bayesian Stackelberg Games}
\label{subsubsec:dynamic_bayesian_stackelberg_games}

In this subsection, we formalize the concept of a \emph{dynamic} Bayesian Stackelberg game. 
Specifically, dynamic Bayesian Stackelberg games generalize static (Bayesian) Stackelberg games by allowing repeated leader-follower interactions where the agent responds non-myopically to maximize their cumulative utility throughout the repeated interactions. 
Formally, a Bayesian Stackelberg game, $\{R, \Theta, \{C^\theta\}_{\theta \in \Theta}, \bmu\}$, is played repeatedly for $T$ rounds. 
The follower has a fixed private type $\theta$ drawn from $\bmu$.
We denote a specific dynamic Bayesian Stackelberg game as $\{R, \Theta, \{C^\theta\}_{\theta \in \Theta}, \bmu, T\}$.
The leader plays a \emph{Dynamic Bayesian Stackelberg Policy}\endnote{Following conventions in sequential decision making, we use ``policy'' to denote a dynamic scheme that maps history to a strategy, whereas ``strategy'' is only used for one-round interactions.} (\texttt{DSP}) $\pi$  which specifies a leader strategy $\bx^t  = \pi(\bj_{t-1}) \in \Delta^m$ at each round $t$, where $\bj_{t-1} = (j_1, \cdots, j_{t-1})$\endnote{We use bold $\bj_t$ to denote a vector throughout the paper, where the subscript $t$ represents the vector's length.} is the follower's past responses. The leader commits to a \texttt{DSP}  $\pi$ before the game starts and the follower observes the policy in advance.
This commitment assumption is common in the literature, e.g. the dynamic pricing problem \citep{devanur2014perfect,vanunts2019optimal}, dynamic mechanism design \citep{amin2013learning,pavan2014dynamic,mirrokni2020non}  and Stackelberg security games \citep{sinha2018stackelberg}.
We call the optimal \texttt{DSP} the \emph{Dynamic Bayesian Stackelberg Equilibrium} ($\DSE$).

\subsubsection{Learning in Dynamic Bayesian Stackelberg Games}
\label{subsubsed:learning_in_dynamic_bayesian_stackelberg_games}

The leader in a dynamic Bayesian Stackelberg game is not primarily concerned with identifying the follower type, as in  traditional learning paradigms.
Instead, the leader seeks to maximize her utility over the rounds of the game. However, to do this, the leader may take advantage of the follower's \emph{revealed preferences} \citep{beigman2006learning,roth2016watch}, so that she can distinguish the follower type hence tailor her strategies accordingly.  
As a byproduct of such optimization, a dynamic policy $\pi$ may learn  the follower type.
From this perspective,    $\DSE$ as the optimal dynamic Stackelberg policy is a concept that combines learning with utility optimization. 
Formally, we define effective learning as follows:
\begin{definition}[Effective Learning]
 For a dynamic Bayesian Stackelberg game $\{R, \Theta, \{C^\theta\}_{\theta \in \Theta},$ $ \bmu, T\}$, a \texttt{DSP} strategy  $\pi$ \emph{learns effectively} if there exists $\theta,\theta' \in \Theta$ and $t \in \{2,...,T\}$ such that for follower best response histories $\bj^{*\theta}_{t-1}$ and $\bj^{*\theta'}_{t-1}$, $\pi(\bj^{*\theta}_{t-1}) \ne \pi(\bj^{*\theta'}_{t-1})$ and the leader utility realized by $\pi$ exceeds the leader utility realized by repeatedly playing the \texttt{BSE} leader strategy for $T$ rounds.
\end{definition}


Note that the $\DSE$ is, by necessity, a dynamic policy, and we assume that the leader has full commitment power.
Therefore, even with full knowledge of the follower type, the $\DSE$ should be better than the optimal static strategy due to the commitment power of the leader.
However, we only consider learning to be effective when the leader employs different strategies based on divergent best response histories.
Strategies that adapt to distinct best response histories require learning, in addition to commitment.
In our experiments, we generally do not, and cannot, distinguish between improvements relative to static policies due to commitment versus learning.
However, in Section~\ref{sec:whyDSElearns} we demonstrate that learning alone will ensure higher utility for the leader, relative to the optimal static policy, in most dynamic Bayesian Stackelberg games.

\subsubsection{Dynamic Pricing Games}
\label{section:dynamic_pricing_games}

The key motivation for this work is the \emph{No Learning Theorem} for the dynamic pricing game with a single buyer with fixed valuation, i.e. the \emph{dynamic pricing problem}.
In this section, we formally define the dynamic pricing problem.
In this game, a seller (she) repeatedly sells an item to the same buyer (he) for $T$ rounds.
The buyer has a  private value $v \in \mathbb{R}_+$ for the item, which is drawn from some prior distribution $\bmu$ before the game starts and is fixed throughout the game. 
The seller knows the prior distribution $\bmu$ and can post a price $p_t$ at each round $t$; the buyer responds with $j_t \in \{0, 1\}$, indicating accepting the price and buying the item ($j_t = 1$) or not. 
The buyers value is $j_t \cdot(v - p_t)$ for round $t$.
Before this game starts, the seller commits to a \emph{dynamic pricing policy} $\pi$
that maps any buyer's past responses to a price at the current round $t$, i.e. $p_t = \pi(j_1, \cdots, j_{t-1} )$. 

One potential policy is to repeatedly post the Myerson price \citep{myerson1981optimal}, $p^* = \arg \max_{p} [p (1-\text{Pr}(v\le p))]$, which maximizes the single-round revenue under seller's prior knowledge $\bmu$.
This policy would appear highly sub-optimal -- for example, if the buyer rejects the item in the first round, he will reject it in all future rounds, leading to a revenue of zero.
One might conjecture that the seller should be able to gradually learn the exact $v$ from repeatedly observing the buyer's responses, and then set the price to $v$, extracting maximum revenue in all future rounds.
However, the no learning theorem shows that, due to the buyer's strategic responses to the seller's learning, it is \emph{impossible} for the seller to achieve higher expected revenue using any other strategy \citep{vanunts2019optimal}!

\begin{theorem}[No Learning Folk Theorem for the Dynamic Pricing Game]
    The leader utility realized by the repeated $\BSE$ is identical to the leader utility realized by the $\DSE$.
\end{theorem}

It is important to note that in the pricing game, the seller does \emph{learn} about the buyer's type under the optimal, static pricing, policy.
Specifically, the seller will know whether or not $v \ge p^*$.
However, the seller cannot condition the pricing policy on the history of responses, so the seller does not learn \emph{effectively}.
Refer to Appendix~\ref{appendix_sec:no_learning_pricing} for a proof of the No Learning Theorem.

\section{Learning in Dynamic Bayesian Stackelberg Games}
\label{sec:whyDSElearns} 

Dynamic pricing is a special case of a general dynamic Bayesian Stackelberg game with the seller as the leader and the buyer as the follower.
Given the optimality of static pricing due to the No Learning Theorem, it becomes very natural to ask the following research question:
\begin{quote}
   {
\it Is it possible to learn effectively from a strategic follower in general dynamic Bayesian Stackelberg games?
}
\end{quote}
As the following examples illustrate, it is indeed possible to learn effectively from a strategic follower, even when the dynamic Stackelberg game is a small modification of a dynamic pricing game.

\begin{example}
    \label{example:example_1}
{\it     Let a Bayesian Stackelberg game be defined by the payoffs listed in Table \ref{table:DySS_better_MRSS}, and assume that types $0$ and $1$ are equally likely.
    As can be computed, the optimal \emph{static} policy for the leader in this game is the (\BSE) leader strategy  $\bx = (\frac{1}{3}, \frac{2}{3})$, i.e. the leader commits to a randomized strategy which plays action $i_0$ with probability $\frac{1}{3}$ and $i_1$ with probability $\frac{2}{3}$. This $\BSE$ strategy results in  an expected leader utility of $5\frac{1}{6}$. } 
    \begin{table}
    \begin{minipage}{0.31\linewidth}
        \centering
        \begin{tabular}{|c|c|c|} \hline
          $R$  &  $j_0$ & $j_1$  \\ \hline
          $i_0$   & 5 & 2  \\  \hline
          $i_1$   & 5 & 7 \\  \hline
        \end{tabular} 
    \end{minipage}
    \begin{minipage}{0.33\linewidth}
        \centering
        \begin{tabular}{|c|c|c|} \hline
          $C^0$  &  $j_0$ & $j_1$ \\  \hline
          $i_0$   & 5 & 2  \\  \hline
          $i_1$   & 4 & 2 \\  \hline
        \end{tabular} 
        \end{minipage}
        \begin{minipage}{0.33\linewidth}
        \centering
        \begin{tabular}{|c|c|c|} \hline
          $C^1$  &  $j_0$ & $j_1$  \\ \hline
          $i_0$   & 5 & 7  \\  \hline
          $i_1$   & 4 & 3 \\  \hline
        \end{tabular} 
        \end{minipage}
        \caption{Utility Matrices for Example~\ref{example:example_1}. A Bayesian Stackelberg game instance with two actions for the leader and follower. The leader's utility matrix is the first subtable $R$; the follower has two possible types, with utility matrix $C^0, C^1$ respectively, each having equal probability $0.5$. For readability, we relax the assumption that all utilities are in $[0,1]$.}
        \label{table:DySS_better_MRSS}
    \end{table}
\end{example}
         
In Example~\ref{example:example_1}, consider the simplest possible dynamic setup with \emph{two} rounds of leader-follower interactions. The optimal dynamic leader policy is to play $\bx^1=(0, 1)$ in the first round. If the follower responds with $j_0$ in the first round, the leader still plays $\bx^2 = (0,1)$ in the second round; otherwise, the leader plays  $\bx^2 = (1/2, 1/2)$. 
Follower  $C^0$ best responds by playing $j_0$ in both rounds, i.e. $\bj^{*0}_2 = (j_0, j_0)$, while $C^1$ plays $j_1$ in both rounds, i.e. $\bj^{*1}_2 = (j_1, j_1)$.
This implies that the leader plays different strategies against the two types in the second round, i.e $\pi(\bj^{*0}_{1}) \ne \pi(\bj^{*1}_{1})$; the optimal policy requires that the leader learns, and exploits knowledge of, the followers type.
Moreover, this dynamic policy results in a total leader utility of $10\frac{3}{4}$, averaged to $5 \frac{3}{8}$  per round, which is strictly larger than the static optimal leader utility of $5 \frac{1}{6}$.

The optimal dynamic policy also outperforms the leader's static strategy with complete knowledge.
Specifically, suppose the leader can observe the follower's type before the game and can play the optimal static strategy against each type (i.e. the $\SSE$), which is $\bx=(1, 0)$ against $C^0$ and $\bx=(\frac{1}{3},\frac{2}{3})$ against $C^1$.
In this case, the leader obtains expected utility $5$ against follower type $C^0$ and $5\frac{1}{3}$ against $C^1$, both of which are less than the averaged dynamic utility $5\frac{3}{8}$.
This is because dynamic strategies, with commitment, are intrinsically more powerful than static strategies.
Specifically, the leader is able to induce type $1$ to play $j_1$ in the first round, when he would prefer to play $j_0$, because the follower is rewarded in the second round with the more favorable leader strategy of $\bx^2 = (1/2, 1/2)$.
Stated another way, even if the leader knew the followers type, a dynamic policy would outperform the optimal static policy due to commitment.
In this example, both learning and commitment contribute to the increase in the leader's utility.

While Example~\ref{example:example_1} demonstrates that learning can be effective in a Bayesian Stackelberg game, this is certainly not guaranteed.
We can reformulate the dynamic pricing game as a Bayesian Stackelberg game, as Example~\ref{example:pricing_game_example} illustrates in the following, and then the No Learning Theorem implies that no dynamic policy can outperform the optimal static policy.


\begin{example}[Pricing as a Stackelberg Game]
    \label{example:pricing_game_example}
{\it     Consider the standard pricing problem where the seller sells a single item to a buyer  whose value $v$ over the item is drawn from a distribution which, as an example here, is the uniform distribution over set $  V=\{8, 35, 96\}$.
    The seller will set prices at a possible value from set $V$.
    This game is  then equivalent to a Bayesian Stackelberg game with utility matrices described in Table~\ref{table:example_pricing}.
    It is easy to compute that the optimal Myerson price, and therefore the optimal static policy, in this setting is $96$, which induces expected seller  revenue  $\mathbf{32}.$ }
    \begin{table}
    \begin{minipage}{0.24\linewidth}
        \centering
        \begin{tabular}{|c|c|c|} \hline
          $R$  &  $j_0$ & $j_1$  \\ \hline
          $i_0$   & 0 & 8 \\ \hline
          $i_1$   & 0 & 35 \\ \hline
          $i_2$   & 0 & 96 \\ \hline
        \end{tabular} 
    \end{minipage}
    \begin{minipage}{0.24\linewidth}
        \centering
        \begin{tabular}{|c|c|c|} \hline
          $C^0$  &  $j_0$ & $j_1$  \\ \hline
          $i_0$   & 0 & 0 \\ \hline
          $i_1$   & 0 & -27 \\ \hline
          $i_2$   & 0 & -88 \\ \hline
        \end{tabular} 
        \end{minipage}
    \begin{minipage}{0.24\linewidth}
        \centering
        \begin{tabular}{|c|c|c|} \hline
          $C^1$  &  $j_0$ & $j_1$  \\ \hline
          $i_0$   & 0 & 27 \\ \hline
          $i_1$   & 0 & 0\\ \hline
          $i_2$   & 0 & -61\\ \hline
        \end{tabular} 
    \end{minipage}
    \begin{minipage}{0.24\linewidth}
        \centering
        \begin{tabular}{|c|c|c|} \hline
          $C^2$  &  $j_0$ & $j_1$  \\ \hline
          $i_0$   & 0 & 88 \\ \hline
          $i_1$   & 0 & 61 \\ \hline
          $i_2$   & 0 & 0 \\ \hline
        \end{tabular} 
        \end{minipage}
    \caption{Utility Matrices for Example~\ref{example:pricing_game_example}. A pricing game example where $R$ represents the seller's utility matrix, and $C^0$, $C^1$, $C^2$ represents the buyer's utility matrices when his type is $v=8$, $v=35$, and $v=96$ correspondingly. In addition, $i_0, i_1, i_2$ represents setting a price of $8, 35, 96$, while $j_0/j_1$ represents the buyer rejects/accepts the price. For readability, we relax the assumption that all utilities are in $[0,1]$.}
    \label{table:example_pricing}
    \end{table}
\end{example}
    
It can be verified using the approach we propose in Section \ref{sec:algo} that the optimal dynamic seller policy is to set the Myerson price $96$ at every round, which is the optimal static price and is consistent with the No Learning Theorem for the pricing game.
However, we can modify the pricing game slightly so that the seller's (leader's) action space does not include prices that correspond to the buyer's (follower's) set of valuations.
This small modification to the pricing game leads to a significantly different optimal dynamic strategy.


\begin{example}[Modified Pricing Game]
    \label{example:mod_pricing_game}
 {\it    Define a modified pricing game, similar to Example~\ref{example:pricing_game_example}, where the buyers valuations are again in the set $V=\{8, 35, 96\}$
    However, assume that the seller can only offer the set of prices ${22, 40, 61}$.
    Given these restricted prices, the seller/leader's and the buyer/follower's utility matrices are defined as in Table~\ref{table:example_pricing_restricted}.} 
    
    \begin{table}
    \begin{minipage}{0.24\linewidth}
        \centering
        \begin{tabular}{|c|c|c|} \hline
          $R$  &  $j_0$ & $j_1$  \\ \hline
          $i_0$   & 0 & 22 \\ \hline
          $i_1$   & 0 & 40 \\ \hline
          $i_2$   & 0 & 61 \\ \hline
        \end{tabular} 
    \end{minipage}
    \begin{minipage}{0.24\linewidth}
        \centering
        \begin{tabular}{|c|c|c|} \hline
          $C^0$  &  $j_0$ & $j_1$  \\ \hline
          $i_0$   & 0 & -14 \\ \hline
          $i_1$   & 0 & -32 \\ \hline
          $i_2$   & 0 & -53 \\ \hline
        \end{tabular} 
        \end{minipage}
    \begin{minipage}{0.24\linewidth}
        \centering
        \begin{tabular}{|c|c|c|} \hline
          $C^1$  &  $j_0$ & $j_1$  \\ \hline
          $i_0$   & 0 & 13 \\ \hline
          $i_1$   & 0 & -5 \\ \hline
          $i_2$   & 0 & -26 \\ \hline
        \end{tabular} 
        \end{minipage}
    \begin{minipage}{0.24\linewidth}
        \centering
        \begin{tabular}{|c|c|c|} \hline
          $C^2$  &  $j_0$ & $j_1$  \\ \hline
          $i_0$   & 0 & 74 \\ \hline
          $i_1$   & 0 & 56 \\ \hline
          $i_2$   & 0 & 35 \\ \hline
        \end{tabular} 
    \end{minipage}
    \caption{Utility Matrices for Example~\ref{example:mod_pricing_game}. A pricing game example where $R$ represents the 
    leader's utility matrix, and $C^0$, $C^1$, $C^2$ represents 
    the follower's utility matrices when his type is $v=8$, 
    $v=35$, and $v=96$ correspondingly. In addition, $i_0, i_1, i_2$ represents setting a price of \textit{22, 40, 61} 
    (note the difference from table \ref{table:example_pricing}), while $j_0/j_1$ represents the 
    follower rejects/accepts the price. For readability, we 
    relax the assumption that all utilities are in [0,1].}
    \label{table:example_pricing_restricted}
    \end{table}
\end{example}


Interestingly, with this small modification to the pricing game, a dynamic policy indeed outperforms the optimal static policy. The optimal static leader policy can be computed as $\left(\frac{2}{3},0, \frac{1}{3}\right)$, i.e. setting an expected price of $35$, which gives the leader an expected average utility of $\mathbf{23\frac{1}{3}}$. 
On the other hand, the optimal dynamic policy is to play a mixed strategy $(0, 0, 1)$ (i.e. set a price of $61$) in the first round. If the price is rejected, the leader switches to the mixed strategy of $\left(\frac{2}{3},0, \frac{1}{3}\right)$ (i.e. set a price of $35$); the leader keeps the same price of $61$ if the price is accepted in the first round. 
As a result, Follower $C^0$ with private value $8$ rejects for two rounds; Follower $C^2$ with private value $96$ accepts the price for two rounds; Follower $C^1$ rejects the price in the first round and then accepts in the second round. 
This dynamic policy gives the leader an expected average utility of $\mathbf{26\frac{1}{6}}$, showing that the optimal dynamic policy outperforms the optimal static policy. 
 
This example also provides some intuition about one question an attentive reader may ask, ``Given that the allocation rule gives the seller more power than the leader in a general Stackelberg game, why is the seller less able to learn than in the general Stackelberg game?''
Notably, in the above example, the optimal dynamic utility is still worse than Myerson's revenue of $\mathbf{32}$ when the seller has an available price of $96$.
Therefore, even though the no learning theorem breaks down in the Stackelberg game (Table \ref{table:example_pricing_restricted}) in the sense that the optimal dynamic policy outperforms the optimal static policy, it does not outperform the Myerson revenue.

Additionally, Example~\ref{example:mod_pricing_game} does not rely on commitment directly to improve the seller's utility.
In this case, the seller offers the optimal (restricted) price assuming that the buyer is the highest type, and if the buyer does not buy, she offers the optimal price for the middle type.

\subsection{Sufficient Condition for Learning in Dynamic Bayesian Stackelberg Games}
\label{subsec:sufficient_condition_learning_random_games}

In the preceding examples, we have demonstrated that learning can be effective in dynamic Bayesian Stackelberg games.
However, the question remains, ``Is learning generally effective for dynamic Bayesian Stackelberg games?''
In this section and the next, we formally show that, yes, learning is generally effective for dynamic Bayesian Stackelberg games in contrast to the No Learning Theorem.
In order to accomplish this, we characterize a sufficient condition to ensure that learning is effective.

Interestingly, the proof, see Appendix~\ref{subappend:sufficient_condition_learning_random_games}, demonstrating that the condition given in Assumption~\ref{assumption:learnable_subgroup} is indeed sufficient is constructive, and it purely relies on learning, not commitment, to improve relative to the optimal static policy.
Specifically, the strategy that we will employ will be to identify a subset of types to target for the dynamic policy.
Starting from the \texttt{BSE}, the leader chooses to deviate to a strategy that improves relative to the \texttt{BSE} in the last round if and only if the agent responds consistently as the chosen subset.
All types of agents are incentivized to respond with their $\BSE$ best response in each round due to it being costly to emulate another type.
We show that given some assumptions, this policy strictly improves relative to the \texttt{BSE}.

Again, this constructed policy switches from a \texttt{BSE} for the whole group to a \texttt{BSE} for a subset of the types.
In both cases, the commitment power does not increase the utility relative to purely learning.
Stated differently, the leader switches to an optimal static strategy after being sufficiently convinced that the follower's type belongs to a subset of the types, though she does commit to only exploiting this information in the final round which is essential to prevent other types pooling with this sub-group of types.
Therefore, commitment is essential to the learning strategy, but learning is sufficient to ensure a higher utility for the leader.

We will make use of the notion of a sub-group \texttt{BSE}, defined as follows:
\begin{definition}[Sub-group \texttt{\BSE}]\label{def:sub-group_BSE}
    Given a Bayesian Stackelberg game $\{R, \Theta, \{C^\theta\}_{\theta \in \Theta}, \bmu\}$, consider the sub-group of types $\Theta' \subset \Theta$ and re-normalized distribution $\mu'(\theta) = \frac{\mu(\theta)}{\sum_{\theta' \in \Theta'} \mu(\theta)}$ for $\theta \in \Theta'$.
    The \emph{sub-group \texttt{BSE}} for $\Theta'$, denoted $\texttt{BSE}(\Theta')$, is the \texttt{BSE} for the Bayesian Stackelberg game $\{R, \Theta', \{C^\theta\}_{\theta \in \Theta'}, \bmu'\}$.
\end{definition}

Note that for a sub-group of size one, i.e. $\{\theta\}$ for some $\theta \in \Theta$, $\texttt{BSE}(\{\theta\})$ would correspond to the \texttt{SSE} for type $\theta$.
Additionally, we will be concerned with both the sets of best responses for a given set of types assuming a certain leader strategy as well as the set of static leader strategies that induce a certain best response for each type.
We define the notation for these sets as follows:

\begin{definition}[Best Response Set]\label{def:best-response-set}
For a given leader strategy $\x$, the \emph{set of best 
responses} for a sub-group $\Theta' \subset \Theta$ is all 
follower actions that achieve maximal utility for some 
follower in subset $\Theta'$, i.e., 
\begin{align*}   
    \texttt{BR}(\Theta', \x) = \Big\{j^* \in [n] \,\, \big| \,\,\exists \theta \in \Theta'  
    \textnormal{ such that } j^* \in \argmax_{j\in [n]} V^\theta(\x, j)\Big\}.
\end{align*}
Similarly, for a given type $\theta \in \Theta$ and action $j\in [n]$, we define the region of the leader strategy space for which $j$ is the best response for type $\theta$ as $\texttt{BR}_\theta(j)$, i.e.
\begin{equation*}
    \texttt{BR}_\theta(j) = \left\{\x \in \Delta^m \mid j \in \texttt{BR}(\theta, \x) \right\}.
\end{equation*}
\end{definition}

Now we can state our sufficient condition to ensure that learning is effective for a dynamic Bayesian Stackelberg game.

\begin{assumption}[Existence of Learnable Sub-group]\label{assumption:learnable_subgroup}
    Given \texttt{BSE} leader strategy $\boldsymbol{x}^*$, there exists a sub-group $\Theta' \subset \Theta$ such that $\BSE(\Theta') \ne \x^*$ and $\texttt{BR}(\Theta', \boldsymbol{x}^*) \cap \texttt{BR}(\Theta \setminus \Theta', \boldsymbol{x}^*) = \emptyset$.
\end{assumption}

As we will show in Theorem~\ref{thm:effective_learning_thm}, Assumption~\ref{assumption:learnable_subgroup} is a sufficient condition to ensure that there exists a policy that learns effectively.
Assumption~\ref{assumption:learnable_subgroup} states that there exists a subgroup whose set of best responses is entirely disjoint from the rest of the follower types.
Unsurprisingly, the pricing game does not satisfy Assumption~\ref{assumption:learnable_subgroup}.
In the pricing game, at the $\BSE$, there is always one type that is indifferent between purchasing and not purchasing.
Therefore, there is no best response set that is disjoint from the other types.

It may seem a-priori reasonable that a strategic leader should be able to induce a subgroup with distinct best responses to eventually reveal their type.
However, it is not obvious that this could be done in such a manner that the leader is made better off versus the optimal static strategy, which we will show is indeed the case.
More importantly, given the precise condition of Assumption~\ref{assumption:learnable_subgroup} that ensures that effective learning is possible, we can leverage the assumption to show that for random games, according to a certain definition of random, learning is very likely to be effective, as we will demonstrate in Section~\ref{subsec:learning_random_games}. 
This strongly suggests that we should not expect a No Learning Theorem type result to hold for an arbitrary class of Bayesian Stackelberg games.
Instead,   insights from the No Learning Theorem should   be viewed as narrowly applicable to the dynamic pricing game.

Finally, Assumption~\ref{assumption:learnable_subgroup} is a sufficient condition, not a necessary condition.
For example, the pricing game in Example~\ref{example:mod_pricing_game} does not satisfy Assumption~\ref{assumption:learnable_subgroup}.
 Therefore, any estimation of the proportion of games for which Assumption~\ref{assumption:learnable_subgroup} holds will be an upper bound on the true number of games for which learning is effective. The proof of the theorem can be found in Appendix~\ref{subappend:sufficient_condition_learning_random_games}.

\begin{theorem}\label{thm:effective_learning_thm}
    Given a Bayesian Stackelberg game $\{R, \Theta, \{C^\theta\}_{\theta \in \Theta}, \bmu\}$ that satisfies assumption \ref{assumption:learnable_subgroup}, there exists a $T^*$ such that for all $T \ge T^*$, the dynamic Bayesian Stackelberg game $\{R, \Theta, \{C^\theta\}_{\theta \in \Theta}, \bmu, T\}$ admits a $\texttt{DSP}$, $\pi$, that learns effectively.
\end{theorem}

\subsection{Effective Learning in Random Bayesian Stackelberg Games}
\label{subsec:learning_random_games}

While Theorem~\ref{thm:effective_learning_thm} indicates that Assumption~\ref{assumption:learnable_subgroup} is sufficient to ensure that learning is effective, the primary question is whether or not Assumption~\ref{assumption:learnable_subgroup} is reasonable.
In this section, we demonstrate that for a randomly generated Bayesian Stackelberg game the probability that effective learning is possible is increasing in the number of follower types, given some mild assumptions on the game generation process.
This suggests that effective learning being possible is, in some sense, the most likely state.
Moreover, our definition of a random Bayesian Stackelberg game only assumes randomness over the \emph{followers} payoffs.
For the leader, we simply need the condition that there is not a dominant action, i.e. an action for which, no matter the followers response, is at least as good as any other action.
Clearly, learning cannot be effective if there is a dominant action, so this demonstrates that it is sufficient for the follower's payoff matrices to be ``general'' in order to imply that learning is effective with high probability.

For this section, we will be primarily be concerned with $(m-1)$-dimensional hyperplanes that define the sets of leader strategies for which a follower of type $\theta \in \Theta$ is indifferent between two actions $i,j \in [n]$.
These hyperplanes are defined by the sets $H_{i,j}^\theta = \{\x \in \mathbb{R}^m \mid \x\cdot(C_{\cdot, i}^\theta - C_{\cdot, j}^\theta) = 0\}$, where we denote by $C^\theta_{\cdot, j}$ the column vector of utilities for the follower of type $\theta \in \Theta$ when he plays action $j$.
We will denote by $G(m,k)$ the Grassmanian of $k$-dimensional linear subspaces of $\mathbb{R}^m$.
Therefore, $H_{i,j}^\theta \in G(m,m-1)$.
Additionally, we will be working with the half spaces of leader strategies for which $i$ is (weakly) preferred to $j$, i.e. $H_{i,j}^{\theta,+} = \{\x \in \mathbb{R}^m \mid \x\cdot(C_{\cdot, i}^\theta - C_{\cdot, j}^\theta) \ge 0\}$.
Note that for the leader strategies to be valid, we must restrict the sets to the intersection with $\Delta^m$.
However, for this section, we will generally not make that explicit.
We will demonstrate the non-existence of leader strategies that satisfy certain conditions, and clearly if a leader strategy does not exist in $\mathbb{R}^m$, it does not exist in $\Delta^m \subset \mathbb{R}^m$.
We will denote a random variable using a calligraphic font, e.g. $\mathcal{H}_{i,j}^{\theta,+}$.

\begin{definition}[Random Bayesian Stackelberg Game]
    A \emph{random Bayesian Stackelberg Game} is a Stackelberg game with leader utility matrix $R$ and prior distribution $\boldsymbol{\mu}$ over a fixed set of types $\Theta$ where the follower utility matrices, $\{\mathcal{C}^\theta\}_{\theta \in \Theta}$, are independently and identically distributed according to $f$. We denote a random game by $\{R, \Theta, \{\mathcal{C}^\theta\}_{\theta\in \Theta},\boldsymbol{\mu},f\}$.
\end{definition}

However, it is clear that an arbitrary distribution $f$ will not provide any non-zero bound on the probability that the random Bayesian Stackelberg game permits effective learning since any probability distribution that strictly generates pricing games does not permit effective learning with probability $1$.
Therefore, we must put some conditions on the generating distribution in order to ensure that the resulting game is likely to permit effective learning.
Our condition is fundamentally that no particular follower action is more likely to be the best response for a given leader strategy across all follower types.
Again, this seems a natural condition, given that if a single response is likely to always be a best response, then the leader cannot benefit from learning the follower's type.

\begin{assumption}[Generic Conditions for Follower Payoffs]
    \label{assum:generic_condition}
    Consider the random Bayesian Stackelberg Game $\{R, \Theta, \{\mathcal{C}^\theta\}_{\theta\in \Theta},\boldsymbol{\mu},f\}$.
    For each $j\in[n]$, there exists a $j'\in[n]$ such that the distribution, denoted $\phi'$, of $(\mathcal{C}^\theta_{\cdot, j} - \mathcal{C}^\theta_{\cdot, j'}) \in \mathbb{R}^m$ induced by $f$, is even and assigns measure zero to a linear subspace of $\mathbb{R}^m$. I.e., for $X \subset \mathbb{R}^m$ and $-X = \{x: -x\in X\} \subset \mathbb{R}^m$, $\phi'(X) = \phi'(-X)$, and for any linear subspace $L \subset \mathbb{R}^m$, $\phi'(L) = 0$.
    
\end{assumption}

Note that $\phi'$ in Assumption~\ref{assum:generic_condition} induces a distribution, which we will denote $\phi$, over $\mathcal{H}_{i,j}^\theta$ and $\mathcal{H}_{i,j}^{\theta, +}$.
The technical result that Assumption~\ref{assum:generic_condition} provides is that it ensures that for all $j\in [n]$ the random hyperplanes $\{\mathcal{H}_{j,j'}^{\theta}\}_{\theta \in \Theta}$ are in \emph{general position} with probability 1 \citep{hugRandomConicalTessellations2016}.
A set of $k$-dimensional hyperplanes $H_1, H_2, ..., H_n \in G(m,m-1)$ are in general position if for any $k \le m$, the hyperplanes have an intersection of dimension $m-k$.
Another way to state this is to consider the unit normal vector to the hyperplane $H_i^\perp = \{x \in \mathbb{R}^m \mid x\cdot y = 0 \ \forall \ y\in H_i \ \wedge \ ||x||=1\}$.
Then $H_1, H_2, ..., H_n \in G(m,m-1)$ are in general position if for any subset of size $m$ or less of $H_1^\perp, H_2^\perp, ..., H_n^\perp \in \mathbb{R}^m$ the subset is linearly independent.

Assumption~\ref{assum:generic_condition} holds for natural distributions over the set of payoff matrices.
For example, any distribution $f$ such that the columns of the follower's payoff matrix $\mathcal{C}_{\cdot,j}^\theta$ are independent and identically distributed for all $\theta \in \Theta$ and $j \in [n]$ satisfies Assumption~\ref{assum:generic_condition}, assuming that the distribution over columns does not assign positive measure to a linear subspace of $\mathbb{R}^m$.
The second condition in Assumption~\ref{assum:generic_condition} requires that any linear subspace of $\mathbb{R}^m$ is assigned measure zero by $\phi$.
However, in all the following results, this can be relaxed for cases under which the support of the distribution is confined to a linear subspace of dimension $k \le m$ but otherwise assigns measure zero to linear subspaces of dimension $k' \le k$.
This allows for settings under which the follower's payoffs for two or more leader actions are perfectly correlated.
For the sake of exposition, we will assume that $\phi'$ does not assign a positive measure to a linear subspace.
Given these preliminaries, we can now state our main technical result.

\begin{theorem}
    \label{thm:non_intersecting_BR_regions}
    For a fixed $n$ and $m$ and a randomly generated Bayesian Stackelberg game \newline $\{R, \Theta, \{\mathcal{C}^\theta\}_{\theta\in \Theta},\boldsymbol{\mu},f\}$ that satisfies Assumption~\ref{assum:generic_condition}, $|\Theta| - \left\lceil\log_{\frac{3}{4}}\frac{1}{|\Theta|}\right\rceil > 0$, and $\frac{|\Theta|}{2}\ge m - 2$, then
    \begin{equation*}
        P\left(\bigcup_{i \in [n]}\left(\bigcap_{\theta \in\Theta} \texttt{BR}_\theta(i)\right) \ne \emptyset\right) \le O\left(\frac{1}{|\Theta|}\right)
    \end{equation*}
\end{theorem}

Please see Appendix~\ref{subappend:learning_random_games} for the full proof.
Theorem~\ref{thm:non_intersecting_BR_regions} ensures that the second condition of Assumption~\ref{assumption:learnable_subgroup} is satisfied with high probability; there exists a subgroup $\Theta' \subset \Theta$ such that $\texttt{BR}(\Theta',\x^*) \cap \texttt{BR}(\Theta \setminus \Theta',\x^*) = \emptyset$.
Given that, for all $i\in [n]$, the probability that there exists any $\x \in \Delta^m$ such that $i \in \texttt{BR}(\{\theta\}, \x)$ for all $\theta \in \Theta$ is converging to zero as $\Theta$ increases implies that for $\x^*$, the probability that the subgroup $\Theta'$ exists converges to one.

The final step is to ensure that for this subgroup the leader's optimal strategy, $\texttt{BSE}(\Theta') \ne \x^*$.
The following Corollary ensures that this is the case with high probability.
The proof is in Appendix~\ref{subappend:learning_random_games}.
However, we must additionally assume that the leader does not have a strictly dominant action, i.e. there does not exist a row of $R_{i,\cdot}$ of the leader's utility matrix $R$ that dominates every other row.
If there is such a row, learning cannot be effective since the leader should always just play action $i$.

\begin{corollary}
    \label{cor:learning_is_effective_with_high_probability}
    Let $\{R, \Theta, \{\mathcal{C}^\theta\}_{\theta\in \Theta},\boldsymbol{\mu},f\}$ be a random Bayesian Stackelberg game  that satisfies Assumption~\ref{assum:generic_condition} with $\boldsymbol{\mu}$ having full support on $\Theta$.
    Additionally, let the leader's payoff matrix $R \in \mathbb{R}^{m\times n}$ be such that there does not exist a row $R_{i,\cdot}$ of $R$ such that for all $j \in [m]\setminus {i}$, $R_{i,\cdot}$ is element-wise greater than or equal to $R_{j,\cdot}$; i.e., there does not exist a dominant action for the leader.
    Then the probability that Assumption~\ref{assumption:learnable_subgroup} is not satisfied is $O\left(\frac{1}{|\Theta|}\right)$.
\end{corollary}

See Appendix~\ref{subappend:learning_random_games} for the detailed proof of Corollary~\ref{cor:learning_is_effective_with_high_probability}.
Corollary~\ref{cor:learning_is_effective_with_high_probability} demonstrates that for a fixed $n$ and $m$, learning is likely to be effective for games with many follower types.
We demonstrate this empirically in Table~\ref{table:assumption_justify} where we sample random Bayesian Stackelberg games drawn uniformly at random with a uniform prior.
Table~\ref{table:assumption_justify} clearly shows that as the number of follower types increases, the probability for a game to permit effective learning increases, as expected.
However, it also suggests that the likelihood of learning being effective should increase as the number of follower actions increase.
It is unclear if this is an artifact of the way in which we sample the random games or if this is a more general phenomena.
While we demonstrate that this pattern holds for the case where we sample games from a multivariate normal distribution (see Table~\ref{table:assumption_justify_2} in Appendix~\ref{subsection:effective_learing_additional_simulation_results}), we leave to future work the generality of this phenomenon.

The effect of the number of \emph{leader} actions appears to be more ambiguous.
Intuitively, it seems reasonable that as the number of leader actions increases, learning should be less likely to be effective since the leader has a larger action space from which to construct optimal static strategies.
However, the simulation results in Table~\ref{table:assumption_justify} are inconclusive as to the overall effect.

\begin{table}
    \begin{minipage}{0.24\linewidth}\resizebox{\textwidth}{!}{
    \begin{tabular}{|c|c|c|c|} 
 \hline
   $|\Theta|=2$ & $n=5$  & $n=10$  & $n=15$\\ 
 \hline
 $m=5$  & 40/100  & 56/100& 73/100\\
 \hline
 $m=10$  & 33/100  & 67/100& 66/100\\
 \hline
 $m = 15$  & 40/100   & 62/100& 68/100 \\
 \hline
\end{tabular}}\end{minipage}
    \begin{minipage}{0.24\linewidth}\resizebox{\textwidth}{!}{
    \begin{tabular}{|c|c|c|c|} 
 \hline
   $|\Theta|=3$ & $n=5$  & $n=10$  & $n=15$\\ 
 \hline
 $m=5$  & 53/100  & 85/100& 91/100\\
 \hline
 $m=10$  & 57/100  & 83/100 & 85/100\\
 \hline
 $m = 15$  & 52/100   & 79/100 & 89/100\\
 \hline
\end{tabular}}\end{minipage}
    \begin{minipage}{0.24\linewidth}\resizebox{\textwidth}{!}{
    \begin{tabular}{|c|c|c|c|} 
 \hline
   $|\Theta|=4$ & $n=5$  & $n=10$  & $n=15$\\ 
 \hline
 $m=5$  & 65/100  & 90/100& 95/100\\
 \hline
 $m=10$  & 70/100  & 85/100 & 93/100\\
 \hline
 $m = 15$  & 57/100   & 89/100 & 95/100\\
 \hline
\end{tabular}}\end{minipage}
    \begin{minipage}{0.24\linewidth}\resizebox{\textwidth}{!}{
    \begin{tabular}{|c|c|c|c|} 
 \hline
   $|\Theta|=5$ & $n=5$  & $n=10$  & $n=15$\\ 
 \hline
 $m=5$  & 71/100  & 94/100& 98/100\\
 \hline
 $m=10$  & 79/100  & 93/100& 99/100\\
 \hline
 $m = 15$  & 60/100   &  92/100 & 98/100\\
 \hline
\end{tabular}}\end{minipage}
\caption{Probability of Assumption~\ref{assumption:learnable_subgroup} being satisfied under $100$ uniform random bayesian Stackelberg game instances. To generate each instance, we sample every game parameter  $R_{i,j}$ and $C^\theta_{i,j}$  $\forall i \in [m], j\in [n], \theta \in \Theta$ independently from the uniform distribution over $[0,1]$. The prior distribution $\bmu$ is uniform over $\Theta$. Each row represents the number of leader actions, and the column represents the number of follower actions.}
\label{table:assumption_justify}
\end{table}

\section{Learning Versus Communication}\label{sec:utility-compare}

In the preceding section, we demonstrated that learning can be effective for general games, especially when there are numerous possible types and sufficient rounds to facilitate learning.
Furthermore, it is clear that a dynamic policy will always perform at least as well as any static policy in a Stackelberg game, even in scenarios where learning is not feasible.
However, a limitation of the Bayesian Stackelberg framework is that it does not allow for direct communication between the leader and the follower beyond the leader committing to a strategy.
For example, the leader cannot explicitly ask the follower to report their type, as is typically allowed in broader mechanism design settings.
This raises a natural question: is a dynamic strategy simply a substitute for directly eliciting the follower's type?
In other words, is communication between the leader and follower inherently more powerful than learning through repeated interactions?

Although the No Learning Theorem holds for the pricing game even when the leader and follower can communicate, indicating that learning in a dynamic setting generally cannot yield higher utility than direct communication with a static policy (e.g., in pricing games), we demonstrate in this section that the gap is minimal—up to an $O\left(\frac{1}{\sqrt{T}}\right)$ difference.
Specifically, we show that communication is not strictly more advantageous, and that the Dynamic Stackelberg Equilibrium (DSE) can achieve utility that is approximately optimal compared to the static setting with communication.
It is also worth noting that when learning is ineffective in a dynamic setting, it implies that direct communication also offers limited improvement to the leader's utility, again up to an $O\left(\frac{1}{\sqrt{T}}\right)$ factor.
Before presenting our main result for this section, we introduce the concept of the optimal static leader utility in a Bayesian Stackelberg game that allows for communication.






\subsection{Randomized Menu Equilibrium (\texttt{RME})}
By invoking the \emph{revelation principle} of \citet{myerson1982optimal} for general principal-agent problems,  the leader's optimal static strategy is to offer an {\it incentive compatible} (IC) \emph{menu} of randomized strategies. Such a menu can be  described by  $\langle \bm{p}, \bx \rangle =\{p_{\theta,j},\bx_{\theta, j}\}_{j\in [n], \theta \in \Theta}$. After committing to this menu, the leader asks the follower to report his type $\theta$. She then draws  $\bx_{\theta, j}$ with probability $p_{\theta, j}$ and plays $\bx_{\theta, j}$, which will induce follower type $\theta$ to best respond with action $j\in [n]$.
The computation of the optimal menu, which we coin the Randomized Menu Equilibrium ($\RME$),  has been studied recently in Stackelberg games \citep{gan2022optimal} and contract design \citep{castiglioni2022designing}.
\citet{gan2022optimal} shows that the $\RME$ can be computed in polynomial time.

As indicated, this is no longer a Bayesian Stackelberg game equilibrium since it requires both communication between the leader and the follower and that the leader must randomize over a menu of (already randomized) mixed strategies.
This implies the $\RME$ may not be applicable in many settings since: (1) leader-follower communication is not possible in many domains, such as security \citep{sinha2018stackelberg}; (2)  the randomized menu over mixed strategies may not be plausible as a commitment due to the infeasibility of verifying the two layers of randomness.  
Therefore, we treat it mainly as a strong theoretic ``optimality benchmark'' for us to compare against. The main result of this section shows that the $\DSE$ is guaranteed to achieve nearly identical utility, up to a small $O\left(\frac{1}{\sqrt{T}}\right)$ discrepancy. 
Note that we could, alternatively, consider the \textit{dynamic} $\RME$, i.e. the menu of dynamic policies, which would, by the revelation principle, be the strongest possible dynamic policy. 
However, we are concerned with identifying and comparing to the strongest \textit{static} policy that allows communication, so we restrict our attention to the static $\RME$ in this work.

Let   $U^{\DSE}$ and $U^{\RME}$ denote the leader's utility at $\DSE$ and $\RME$. Our result on the comparison between  $U^{\DSE}$ and $U^{\RME}$ hinges on a non-degeneracy assumption pertaining to a notion coined the \textit{inducibility gap}, denoted by $\delta$, a concept adopted in previous work \citep{wuinverse,gan2023robust}. 
\begin{definition}[Inducibility Gap]\label{def:gap}
The inducibility gap of a Stackelberg game with follower types $\Theta$ is the largest $\delta$ such that there exists an IC randomized menu $\{\bx_{\theta,j}, p_{\theta,j}\}_{\theta\in\Theta, j\in[n]}$ where for any follower type $\theta$: $\sum_j p_{\theta,j} V^\theta(\bx_{\theta,j}, j)  \ge \sum_j p_{\theta',j} \max_{j'} V^\theta(\bx_{\theta',j}, j') + \delta, \forall \,\theta' \ne \theta.$
\end{definition}


A positive inducibility gap simply means that every follower type $\theta\in\Theta$ can be \emph{strictly} incentivized to report truthfully by some randomized menu. This is a non-degeneracy assumption since any Bayesian Stackelberg game trivially has $\delta \geq 0$ because playing any fixed mixed strategy $\bx$, irrespective of follower types, already (weakly) incentivizes truthful reports from every follower type. A randomized menu with $\delta$ inducibility gap is called $\delta$-strictly IC.


\begin{theorem}\label{thm:infinte_T}
For any Stackelberg game  with inducibility gap $\delta$ and  $T\geq \Omega(\log^2_n |\Theta|) $, we have $\frac{U^{\DSE}}{T} \geq U^{\RME} - O\left(\sqrt{\frac{\log |\Theta|}{T\delta^2}}\right)$. Moreover, there exist  instances where  $\frac{U^{\DSE}}{T} \leq U^{\RME} - \Omega\left(\frac{1}{T}\right)$.
\end{theorem}

For a proof of Theorem~\ref{thm:infinte_T} see Appendix~\ref{sec:appendix_efficacy}.
Theorem \ref{thm:infinte_T} shows that $\DSE$ achieves equivalent utility to the \RME, up to a $O\left(\sqrt{\frac{1}{T}}\right)$ gap.
It turns out that one can easily find examples in which $\RME$ will be significantly worse --- specifically, $\Omega(1)$ worse --- than the average $\DSE$ utility, even for policies that strictly rely on learning and not commitment.
Consider the following example.

\begin{example}
    \label{example:communication_not_effective}
{\it     Define a Bayesian Stackelberg game with the utility matrices as represented in Table~\ref{table:communication_not_effective} where the probability of each follower type is $\frac{1}{2}$. }
    \begin{table}
    \begin{minipage}{0.31\linewidth}
        \centering
        \begin{tabular}{|c|c|c|} \hline
          $R$  &  $j_0$ & $j_1$  \\ \hline
          $i_0$   & 1 & 0 \\ \hline
          $i_1$   & 0 & 1 \\ \hline
        \end{tabular} 
    \end{minipage}
    \begin{minipage}{0.31\linewidth}
        \centering
        \begin{tabular}{|c|c|c|} \hline
          $C_0$  &  $j_0$ & $j_1$   \\ \hline
          $i_0$   & 0.5 & 0\\ \hline
          $i_1$   & 1 & 0\\ \hline
        \end{tabular} 
    \end{minipage}
    \begin{minipage}{0.31\linewidth}
        \centering
        \begin{tabular}{|c|c|c|} \hline
          $C_1$  &  $j_0$ & $j_1$   \\ \hline
          $i_0$   & 0 & 1\\ \hline
          $i_1$   & 0 & 0.5\\ \hline
        \end{tabular} 
    \end{minipage}
    \caption{Utility Matrices for Example~\ref{example:communication_not_effective}. A Bayesian Stackelberg game instance with two actions for the leader and follower. The leader's utility matrix is the first subtable $R$; the follower has two possible types, with utility matrix $C^0, C^1$ respectively, each having equal probability $0.5$.}
    \label{table:communication_not_effective}
    \end{table}
\end{example}

In Example~\ref{example:communication_not_effective}, the set of $\RME$'s consists of the following menus $\{1, (1,0)\}$ and $\{1, (0,1)\}$, i.e. the leader plays either action $0$ or $1$ with probability $1$.
This optimal static strategy with communication achieves a per round expected utility of $\frac{1}{2}$ or a total utility of $\frac{1}{2}T$ over $T$ rounds.
However, consider the following dynamic strategy without communication:
\begin{align}\label{eq:dynamic_policy}
    \x^1 = (1, 0);  \quad \x^t = \begin{cases} (0, 1) \, \, \text{if $\exists \, t'< t$ and $j_{t'} = j_1$} \\ (1, 0) \, \, \text{otherwise}\end{cases} \text{ for } t>1.
\end{align}
It is straightforward to verify that for type $0$ it is always a best response to play $j_0$ and for type $1$ it is always a best response to play $j_1$.
This leads to an expected first round utility of $\frac{1}{2}$, and in all future rounds, the expected utility is $1$, a total expected utility for the $T$ rounds of $T - \frac{1}{2}$.
Therefore, there is a gap of $\frac{1}{2}(T - 1)$ between the $\RME$ and this dynamic strategy.

Importantly, this increase in utility is driven by learning, not commitment.
Stated differently, if the leader knew the follower's type, the optimal dynamic policy would be a static policy.
In this example, the ability to commit does not increase the leader's utility under full knowledge.
Therefore, it is learning that is driving the increase in utility relative to the $\RME$.

\section{Computing the \texttt{DSE} Exactly and Approximately} \label{sec:algo}

In preceding sections we have demonstrated that in general there exists a dynamic policy which learns effectively for dynamic Bayesian Stackelberg games and that this learning is not merely a substitute for communication.
In this section, we develop algorithms to compute an effective dynamic policy.
However, when $T=1$, the $\DSE$ degenerates to the $\BSE$, which is known to be NP-hard \citep{conitzer2006computing}. 
Moreover, with $T>1$, even writing down a dynamic policy takes space exponential in $T$ as it has to specify a strategy for each possible sequence  of follower actions.
Thus, our focus is on developing practical algorithms for computing the $\DSE$ for small games.
Towards that end, the main result of this section is a Mixed Integer Linear Program ($\MILP$) formulation for computing the  $\DSE$. 
This $\MILP$ has $O(n^T m|\Theta|)$ continuous variables and $O(T|\Theta|)$ integer variables and thus  takes  time exponential in $T, |\Theta|$ to solve. 
The exponential time dependence  on $T, |\Theta|$ is expected since the size of the policy description is $\Omega\left(n^T\right)$ and the exponential-in-$|\Theta|$ time is due to the NP-hardness even when $T=1$.
However, deriving such a $\MILP$ is nontrivial since a naive formulation of the problem is non-linear.
The key technical challenge is deriving a $\MILP$ and demonstrating that, though the program is not strictly equivalent to a program that computes the $\DSE$, any solution of the $\MILP$ can be efficiently converted to a solution for the program that defines the $\DSE$.

The program that maximizes the leader's total utility subject to the constraint that $\bj^{\theta}_{T}$ is the optimal response sequence for follower type $\theta$ is as follows: 
\begin{subequations}
    \label{eq:U^DySS}
    \begin{align}
        &\text{maximize} \quad  \sum_{t \in [T]} \sum_{\theta \in \Theta} \Big[\mu(\theta) \, U(\bx^t_{\bj_{t-1}^\theta}, j^\theta_t) \Big] \\
        &\text{subject to} \quad  \\
        & \sum_{t \in [T]} V^\theta(\bx^t_{\bj_{t-1}^\theta}, j^\theta_t) 
        \geq \sum_{t \in [T]} V^\theta(\bx^t_{\hat{\bj}_{t-1}}, \hat{j}_t), 
       \forall \theta \in \Theta, \, \hat{\bj}_T \in [n]^T, \bj_T^\theta \in [n]^T, \forall \theta \in \Theta. 
    \end{align}
\end{subequations}
Note that in Program \eqref{eq:U^DySS}, the decision variables $\bj_{t-1}^\theta$ appear in the \emph{indices} of other decision variables, $\bx^t_{\bj_{t-1} } $, which means that  Program~\eqref{eq:U^DySS} cannot be solved by   standard solvers.
However, we can construct a $\MILP$ (see \texttt{MILP}~\eqref{eq:DySS_MILP} in Appendix~\ref{append:them1-detail}) that is solvable.

\begin{theorem}\label{thm:program_equivalent}
Program~\eqref{eq:U^DySS} and $\MILP$ \eqref{eq:DySS_MILP} have the same optimal objective value and, moreover, an optimal solution of either can be efficiently recovered from an optimal solution of the other.
\end{theorem}

We defer the proof details to Appendix \ref{append:them1-detail}. Though $\MILP$ \eqref{eq:DySS_MILP}  has $ n T |\Theta|$ binary integer variables $\by$, its running time does \emph{not} depend on $n$ exponentially. This is because any row of $\by^{\theta}$ (i.e. $\by^{\theta}_{t,\cdot}$) has only a single non-zero entry, thus there are in total $O(n^{T|\Theta|})$ many feasible $\by$. Given a feasible $\by$, the program will be an LP with the size of variables bounded by a polynomial in $m$ and $n$ (given constants $T, |\Theta|$). This argument leads to the following corollary of Theorem \ref{thm:program_equivalent}. 
\begin{corollary}
The exact $\DSE$ can be computed by a $\MILP$ with $O(T|\Theta|mn^T)$ continuous variables and $T|\Theta|n$ integer variables; and it can be computed in $\poly(m,n)$ time with constants $T$ and $|\Theta|$.
\end{corollary}

\subsection{Computing the \texttt{DSE}  Approximately}
In this subsection, we present two approximation algorithms for approximating the $\DSE$. 
The first \textsc{Markovian}  approach is a Markovian leader policy, where leader strategy $\x^t$ only depends on the time step $t$ and the follower response $j_{t-1}$.
We can denote the markovian leader policy in round $t$ given $j_{t-1}$ is $\x^t_{j_{t-1}} = \pi(j_{t-1}):[n] \times [T] \rightarrow \Delta^m$.\endnote{At the starting round $t=1$, there does not exist a historical follower response, so the leader plays a starting strategy. We denote it as $\x^1_{j_0}$ for notational consistency.} The specific $\MILP$ can be found in Appendix \ref{sec:approximate_programs}.


\begin{corollary}\label{thm:Markovian_policy}
    The optimal dynamic $\Markov$  policy can be computed by a $\MILP$ with a $O(T|\Theta|mn^2)$ number of continuous variables and $T|\Theta|n$ integer variables.
\end{corollary}
Comparing the $\DSE$, our $\Markov$ policy reduces the exponential number of \textit{continuous} variables to a polynomial size. 
However, despite this improvement, the runtime of the $\Markov$ policy still exhibits exponential growth relative to $T$, making it suboptimal for addressing dynamic leader-follower interactions within a large $T$. 
Therefore, we introduce the following second heuristic approach, \textsc{First}-$k$, a simplified variant of $\DSE$ where the leader's dynamic policy only depends on the initial $k$-rounds of interactions.
Following the initial $k$-rounds, the leader plays a static strategy for the remaining $T-k$ rounds. We defer the specific $\MILP$ to Appendix \ref{sec:approximate_programs}.

\begin{corollary}\label{thm:First_k_policy}
    The optimal dynamic \textsc{First}-$k$  policy can be computed by a $\MILP$ with a $O\big(\hat{k}|\Theta|mn^{\hat{k}}\big)$ number of continuous variables and $|\Theta|n\hat{k}$ integer variables, where $\hat{k} = \min(k+1, T)$.
\end{corollary}

The running time of \textsc{First}-$k$ only grows exponentially in $k$ and the number of follower types $|\Theta|$.
The exponential complexity in $|\Theta|$ is unavoidable due to the inherent hardness of Bayesian Stackelberg games even when $T=1$ \citep{conitzer2006computing}.
Unlike $\DSE$ and $\Markov$ policies, however, the runtime of \textsc{First}-$k$ is independent of $T$. 
As we demonstrate in the following experimental section, \textsc{First}-$k$ appears to scale well with respect to $T$ and achieves higher utility compared to $\Markov$ with a decreased running time.

\section{Experiments}\label{sec:experiment}
To examine the efficacy of the $\DSE$ and the efficiency of our proposed algorithm, we perform several experiments using Gurobi 9.5.1 solver on a machine with Ubuntu 20.04.5 LTS operating system, 2 $\times$ 18 cores 3.0 GHz processors, and 256GB RAM.

\textbf{Game of Chicken.} The \emph{game of chicken} models    situations in which two players desire a shared resource but  conflict arises when both players choose to use the resource simultaneously. 
We consider a variant of this game with a leader-follower structure and two uniformly distributed follower types   $C^0,C^1$ as given in Figure~\ref{fig:game_of_chicken}.
The first row/column action represents ``giving up the resource'' whereas the second row/column action presents ``using the resource''. Type  $C^0$ is the type as in the classic game of chicken whereas $C^1$ is an ``aggressive'' type who strongly prefers using the resource.

\begin{figure}[tbh]
    \centering
    \includegraphics[width=0.75\linewidth]{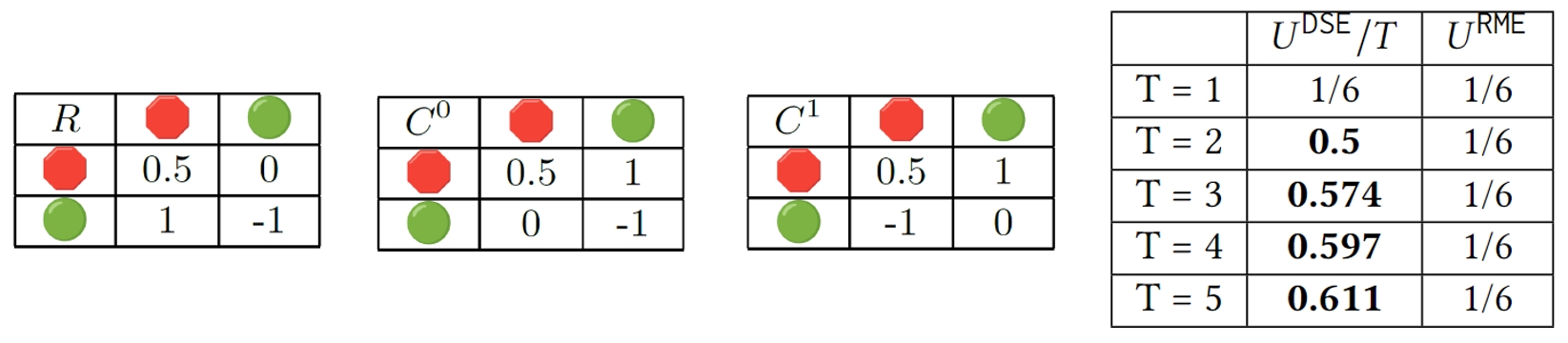}
    \caption{Utility Matrices for the Game of Chicken. A leader-follower variant of the game of chicken with two follower types. The average utility per round for both the $\DSE$ and $\RME$ is shown for $T = 1, ..., 5$.}
    \label{fig:game_of_chicken}
\end{figure}

\textbf{Stackelberg Security Game (SSG).} We consider a specific SSG, where there is a defender (leader) trying to protect 2 targets with 1 resource from the attack of an attacker (follower). 
Specifically, there are two uniformly distributed follower types whose utility information is given by Figure~\ref{fig:stackelberg_security}.
Each row represents the leader's action (i.e. protecting the target) and each column represents the follower's action (i.e. attacking the target). 
Follower type $V^0$ prefers target $t_0$ while $C^1$ prefers target $t_1$.  
\begin{figure}[tbh]
\centering
\includegraphics[width=0.75\linewidth]{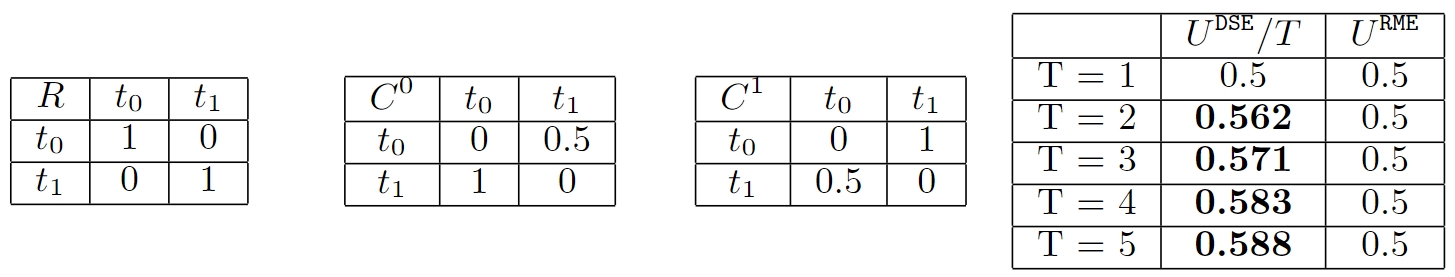}
\caption{A Stackelberg Security Game with two follower types. The average utility per round for both the $\DSE$ and $\RME$ is shown for $T = 1, ..., 5$. Utility Matrices for a Stackelberg Security Game.} 
\label{fig:stackelberg_security}
\end{figure}

\textbf{Battle of the Sexes.}
In this game, two players receive higher utilities when they take the same action, though one of them enjoys this action more than the other. 
We consider a variant of this game with a Stackelberg game structure and two uniformly distributed follower with payoff matrics given by Figure~\ref{fig:battle_of_the_sexes}.
The $C^0$ type is the standard player in the battle of the sexes game, but $C^1$ is a type who ``stubbornly'' insists on the \emph{Basketball} action even though they would enjoy it more if the leader also picks this action. 
\begin{figure}[tbh]
\centering
\includegraphics[width=0.8\linewidth]{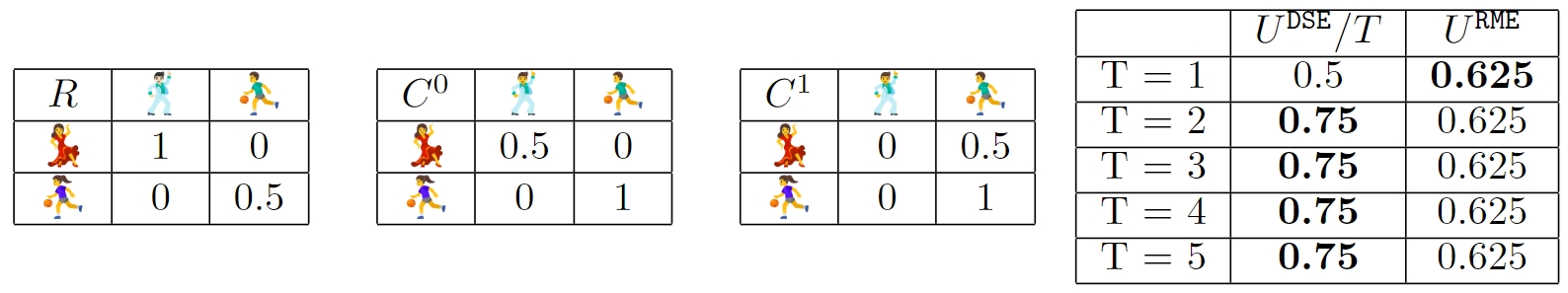}
\caption{Utility Matrices for a Modified Battle of the Sexes Game. A Stackelberg variant of the Battle of the Sexes game with two follower types. The average utility per round for both the $\DSE$ and $\RME$ is shown for $T = 1, ..., 5$.} 
\label{fig:battle_of_the_sexes}
\end{figure}

\textbf{Discussion on Structured Games.} From the experimental results, we can see that the utility improvement of $\DSE$ compared to $\RME$ is evident in all games. 
Note that when $T = 1$, the optimal $\DSE$ is exactly the $\BSE$, which is always less than or equal to $\RME$. 
We also highlight the increasing per round $U^{\DSE}$. 
As shown, though the $U^{\DSE}$ might be below than $U^{\RME}$ in the beginning rounds, it catches up with $U^{\RME}$ as $T$ increases, consistent with Theorem \ref{thm:infinte_T}. 
Interestingly, the $U^\DSE$ converged within two rounds in the battle of sexes game while in the SSG and game of chicken, $U^\DSE$ did not converge within five rounds. 
The convergence of $U^\DSE$ is an interesting open question for future research. 
\textbf{Experimental Results on Randomized Games.} We additionally conducted experiments on random game instances in which each agent's utility is uniformly drawn from $[0,1]$.
As in the previous experiments, we compare the average leader utility between $\DSE$ and the optimal static policy.
In addition, to demonstrate the performance of our approximation algorithms, we also compare the runtime and average leader utility of solving the optimal $\Markov$ policy and \textsc{First}-$k$ policy, versus solving the $\DSE$.

\begin{table}[tb]
\centering
\begin{small}
\begin{tabular}{|c|c|c|c||c|c|c|c|} 
     \hline
     & \multicolumn{3}{c||}{Runtime} & \multicolumn{4}{c|}{Average Utility} \\
     \hline
     & $\First$ & $\Markov$ & $\DSE$ & $\First$ & $\Markov$  & $\DSE$ & $\RME$ \\
     \hline
     T = 1 & $0.11 \pm 0.13$ & $\mathbf{0.1} \pm 0.12$ & $0.11 \pm 0.13$ & $0.9 \pm .05$ & $0.9 \pm .05$& $0.9 \pm .05$ &$\mathbf{0.93} \pm .04$\\
     \hline
     T = 2 & $2.21\pm 1.7$ & $2.34\pm 1.4$ & $\mathbf{2.19} \pm 1.5$ & $0.93 \pm .03$ & $0.93 \pm .03$& $0.93 \pm .03$ &$0.93 \pm .04$\\
     \hline
     T = 3 & $90 \pm 13$ & $\mathbf{8.4} \pm 4.5$ & $91 \pm 13$ & $\mathbf{0.95} \pm .03$ & $0.94 \pm .03$ & $\mathbf{0.95} \pm .03$   & $0.93 \pm .04$\\
     \hline
     T = 4 &$3662 \pm 865$ & $\mathbf{73} \pm 51$  & $3683 \pm 924$ & $\mathbf{0.96} \pm .02$ &$0.95 \pm .02$ &$\mathbf{0.96} \pm .02$  &$0.94 \pm .03$  \\
     \hline
     T = 5 & $4645 \pm 867$ & $\mathbf{2181} \pm 1929$  & N/A & $\mathbf{0.96} \pm .02$ & $0.95 \pm .02$ & N/A & $0.94 \pm .03$ \\
     \hline
     T = 6 & $\mathbf{4614}$  $\pm 564$  & \makecell{$12858$ \\ $\pm  95791$} & N/A & $\mathbf{0.96} \pm .02$ & $0.95 \pm .02$  & N/A & $0.94 \pm .03$ \\
     \hline
     T = 7 & $ \mathbf{4480}$  $\pm 831$  & N/A & N/A & $\mathbf{0.96} \pm .02$ & N/A  & N/A & $0.94 \pm .03$ \\
     \hline
\end{tabular}
\end{small}
\caption{Experimental Results on Randomized Games. Running time (columns 2-4 with the unit: second) and average utility (columns 5-8) for random game instances with $m=10, n=5, |\Theta|=2$, $k=3$. ``N/A'' implies the algorithm did not return a solution within 3 hours.} 
\label{table:m10n5k2_table}
\end{table}
The results from T=$1,2,3$ have been averaged over 50 random instances and the standard deviation is reported in the table. 
As for $T \ge 4$, the runtime is high, and the advantage of \textsc{First}-$k$ and $\Markov$  policy is significant (note the runtime of $\DSE$ when $T=4$ is already much higher than the $\Markov$ method when $T=5$), so we only test the first 10 random instances out of the 50 random instances, which explains why the $\RME$ performs differently when $T\ge 4$. 
In addition, the runtime of $\Markov$ when $T=6$ is already much higher than the \textsc{First}-$k$ method while achieving a lower utility. 
The advantage of $\DSE$ over $\RME$ on the average leader utility is also consistent with our observation in previous structured games. 
Finally, we note that \textsc{First}-$k$ policies achieve near-optimal utility with significantly improved computational efficiency compared to computing the optimal dynamic policy. 
In appendix \ref{appendix_additionalExp}, we report some additional experimental results on both structured and randomized games.
\section{Conclusion}
In this paper, we have shown that, contrary to the intuition suggested by the No Learning Theorem for dynamic pricing, it is generally feasible for the leader to learn effectively in fully strategic settings.
Specifically, the leader can improve her utility through information acquisition in repeated play against a strategic opponent, and this learning is not merely a substitute for communication.

We have also developed a novel mixed-integer linear program to solve for the optimal $\DSE$, albeit with exponential time and space complexity.
To address this computational challenge, we provided two heuristics that perform well and significantly improve computational scalability, though they remain exponential in the number of follower types given that the $\BSE$ is NP-hard to compute.
Taken together, these results suggest that weakening the follower's strategic space to obtain positive learning outcomes, as has commonly been done in the literature, may be unnecessary outside of dynamic pricing settings.

Moving forward, several open questions remain.
The most immediate is whether there exists a simple characterization of a necessary and sufficient condition to ensure effective learning.
While we provide a sufficient condition, it is easy to construct examples, such as Example~\ref{example:mod_pricing_game}, where our sufficient condition is not satisfied, yet learning is still effective.
Additionally, we show that for random Bayesian Stackelberg games, learning is effective with high probability relative to the $\BSE$.
We also demonstrate that learning is at least as powerful as communication (Theorem~\ref{thm:infinte_T}), and by example, we show that learning can be arbitrarily more powerful than communication.
However, we do not have an equivalent result to Corollary~\ref{cor:learning_is_effective_with_high_probability} for the relative performance of the $\DSE$ versus the $\RME$.
Is it also true that for random games the $\DSE$ achieves strictly higher leader utility compared to the $\RME$?
Based on our experimental results, we believe this is very likely, but proving it remains an open question.

\bibliographystyle{ACM-Reference-Format}
\bibliography{refer}

\appendix

\section{Discussion of the No-Learning Theorem in Dynamic Pricing}\label{appendix_sec:no_learning_pricing}
The key technique for proving the optimality of static constant pricing is an elegant reduction from dynamic pricing interactions to a single-round feasible auction mechanism \citep{myerson1981optimal}\footnote{Details about feasible auction mechanisms can be found in Section 3 of \cite{myerson1981optimal}}.
Specifically, for any pricing policy $\pi$, we construct an auction mechanism $M_\pi$ that works as follows. First, the seller asks the buyer to report his value $v$ and simulates $\bj^*$ based on the reported $v$, denoted as:
\begin{align}\label{eq:simulation_jstar}
        \bj^*(v) = \argmax_{\bj \in \{0, 1\}^T} \sum_{t\in [T]} j_t \cdot \big(v - \pi(j_1,\cdots, j_{t-1})\big).
\end{align}
Then the seller randomly chooses a round $t \in [T]$ with probability $1/T$ and allocates the item to the buyer with price $\pi\big(j^*_1(v),\cdots, j^*_{t-1}(v)\big)$ if $j^*_t(v) = 1$. 
\begin{proposition}[Proposition A.1 \citep{vanunts2019optimal}]\label{prop:no-learning-pricing}
The auction mechanism $M_\pi$ is feasible and its expected revenue is equivalent to the dynamic pricing policy $\pi$'s average expected revenue over $T$. 
\end{proposition}
\begin{proof}{Proof.}
To begin with, we formally define the auction mechanism $M_\pi$'s allocation probability $Q_{M}$ and payment $P_M$ with respect to a buyer report $v$.
\begin{align*}
        Q_{M}(v) &= \sum_{t=1}^T \frac{j^*_t(v)}{T}  \\
        P_M(v) &= \sum_{t=1}^T \frac{j^*_t(v) \pi\big(j^*_1(v),\cdots, j^*_{t-1}(v)\big)}{T}
\end{align*}
To prove $M_\pi$ is a feasible auction mechanism, we prove it satisfies incentive compatibility, individual rationality, and $Q_{M}(v) \le 1$. 

First, because $j^*_t(v) \in \{0, 1\}$, we have that $Q_{M}(v) \le 1$. 

Next, we prove the individual rationality of the mechanism, i.e. the buyer's expected utility $v \cdot Q_{M}(v) - P_M(v) \ge 0$. By definition, 
\begin{align}
    v \cdot Q_{M}(v) - P_M(v) =\frac{1}{T} \sum_{t=1}^T j^*_t(v)\big(v - \pi\big(j^*_1(v),\cdots, j^*_{t-1}(v)\big)\big).
\end{align}
In other words, we need to prove $\sum_{t=1}^T j^*_t(v)\big(v - \pi\big(j^*_1(v),\cdots, j^*_{t-1}(v)\big)\big)$, which is exactly the buyer's optimal utility, is greater than or equal to $0$. From the definition of $\bj^*(v)$ in equation  \eqref{eq:simulation_jstar}, this is trivially satisfied, since $0$ is the buyer's utility when setting $j_t = 0, \, \forall t$. 

Finally, we prove the incentive compatibility of the mechanism
\begin{equation}
v \cdot Q_{M}(v) - P_M(v) \ge v \cdot Q_{M}(u) - P_M(u).
\end{equation}
It is equivalent to prove
\begin{align}
    \sum_{t=1}^T j^*_t(v)\big(v - \pi\big(j^*_1(v),\cdots, j^*_{t-1}(v)\big)\big) \ge \sum_{t=1}^T j^*_t(u)\big(v - \pi\big(j^*_1(u),\cdots, j^*_{t-1}(u)\big)\big), 
\end{align}
which is correct by the definition of $\bj^*(v)$ from \eqref{eq:simulation_jstar}.  As a result, we have proved that the mechanism $M_{\pi}$ is a feasible mechanism. 

Last but not least, note that $M_{\pi}$'s expected revenue is the expected payment,
\begin{equation}\label{eq:expected_revenue}
P_M(v) = \frac{1}{T}\sum_{t=1}^T j^*_t(v) \pi\big(j^*_1(v),\cdots, j^*_{t-1}(v)\big)
\end{equation}
where $\sum_{t=1}^T j^*_t(v) \pi\big(j^*_1(v),\cdots, j^*_{t-1}(v)\big)$ is exactly the seller's expected revenue when running dynamic policy $\pi$ for $T$ rounds, proving the proposition.  
\end{proof}

As a result, we have $M_\pi$ as a feasible mechanism for any dynamic pricing policy $\pi$ and the seller's revenue must be bounded by the Myerson revenue \citep{myerson1981optimal}:
\begin{equation*}\label{eq:myerson_price}
   \frac{1}{T} \sum_{t\in[T]}j^*_t(v)\pi\big(j^*_1(v), \cdots, j^*_{t-1}(v)\big) \leq MyersonRev, \forall \pi
\end{equation*}

\noindent Note that $\sum_{t\in[T]}j^*_t(v)\pi\big(j^*_1(v), \cdots, j^*_{t-1}(v)\big)$ is exactly the seller's total revenue in $T$ rounds when she uses dynamic pricing policy $\pi$. Hence, we have shown that \textit{any} dynamic pricing policy's revenue is upper bounded by the revenue of running constant Myerson price for $T$ rounds.

\section{Omitted Details for Section~\ref{sec:whyDSElearns}}
\label{append:detail_learning_is_likely}
In this appendix, we provide the proofs of the results stated in Section~\ref{sec:whyDSElearns}.

\subsection{Proof Details for Section~\ref{subsec:sufficient_condition_learning_random_games}}
\label{subappend:sufficient_condition_learning_random_games}
In this section, we provide the proof of Theorem~\ref{thm:effective_learning_thm}.

\begin{proof}{Proof of Theorem~\ref{thm:effective_learning_thm}.}
We prove this by construction.
Specifically, we construct a dynamic leader policy that achieves higher leader utility than the \texttt{BSE} strategy.  
We do this by constructing a policy that is a local deviation from the \texttt{BSE}. 
We denote the \texttt{BSE} leader strategy as $\x^*$ for $\{R, \Theta, \{C^\theta\}_{\theta \in \Theta}, \bmu\}$. 
By assumption~\ref{assumption:learnable_subgroup}, there exists a sub-group $\Theta' \subset \Theta$ such that $\texttt{BR}(\Theta') \cap \texttt{BR}(\Theta \setminus \Theta') = \emptyset$ and $\texttt{BSE}(\Theta') \ne \x^*$.
Denote $\texttt{BSE}(\Theta')$ as $\hat{\x}$.
We consider the dynamic leader policy, $\pi$, defined as follows:
\begin{align*}
        \x^1 =\cdots = \x^{T-1} = \x^*;  
        \x^T = \begin{cases}
            \hat{\x} \quad \text{if $j_1,\cdots,j^{T-1} \in \texttt{BR}(\Theta')$;} \\
            \x^* \,\,\, \text{otherwise.}
        \end{cases}
\end{align*}
To ensure that all types in $\Theta'$ respond with action in $\texttt{BR}(\Theta')$ for rounds $1, ..., T - 1$ and all types in $\Theta \setminus \Theta'$ do not respond with action in $\texttt{BR}(\Theta')$ for any round, we need the following set of constraints to be satisfied:
\begin{small}
\begin{align}
        &\sum_{\theta \in \Theta'} \mu(\theta) U(\hat{\x}, j^{*\theta}(\hat{\x})) > \sum_{\theta \in \Theta'} \mu(\theta) U(\x^*, j^{*\theta}(\x)) \label{eq:first_constraint_C_types}\\
        &   (T-1) \cdot V^\theta(\x^*, j^{*\theta}(\x^*)) +  V^\theta(\hat{\x}, j^{*\theta}(\hat{\x})) \ge   (T-1)\cdot
        \max_{j \notin \texttt{BR}(\Theta')}V^\theta(\x^*, j) + V^\theta(\x^*, j^{*\theta}(\x^*)), \forall \theta \in \Theta' \label{eq:second_constraint_C_types}\\
        & T \cdot V^\theta(\x^*, j^{*\theta}(\x^*)) \ge    (T-1) \cdot
        \max_{j \in \texttt{BR}(\Theta')} V^\theta(\x^*, j) + V^\theta(\hat{\x}, j^{*\theta}(\hat{\x})), \forall \theta \in \Theta\setminus\Theta' \label{eq:third_constraint_C_types}
\end{align}
\end{small}
Constraint~\ref{eq:first_constraint_C_types} ensures that the leader improves her utility by switching from $\x^*$ to $\hat{\x}$ in the final round of play.
This constraint is satisfied by the assumption that $\texttt{BSE}(\Theta') \ne \texttt{BSE}$.

Constraint~\eqref{eq:second_constraint_C_types} ensures that all types in $\Theta'$ find it optimal to respond with $j^{*'}$ for rounds $1, ..., T - 1$.
To see that this holds, first note that the following must be true by the definition of  $\texttt{BR}(\Theta')$ 
\begin{equation*}
    V^\theta(\x^*, j^{*\theta}(\x)) > \max_{j \notin \texttt{BR}(\Theta')}V^\theta(\x^*, j), \quad \forall \theta \in \Theta'.
\end{equation*}
Then, equation~\eqref{eq:second_constraint_C_types} will be satisfied if the following holds:
\begin{equation*}
    (T-1)\cdot (V^\theta(\x^*, j^{*\theta}(\x^*)) - \max_{j \notin \texttt{BR}(\Theta')}V^\theta(\x^*, j))\ge 1
\end{equation*}

Similarly, constraint~\eqref{eq:third_constraint_C_types} ensures that all types not in $\Theta'$ do not find it optimal to deviate and pretend to be in $\Theta'$.
Note that by Assumption~\ref{assumption:learnable_subgroup} for all $j \in \texttt{BR}(\Theta')$ we have $j \notin \texttt{BR}(\Theta\setminus \Theta')$,
\begin{equation*}
    V^\theta(\x^*, j^{*\theta}(\x^*)) > \max_{j \in \texttt{BR}(\Theta')} V^\theta(\x^*, j)
\end{equation*}
Similarly, equation~\eqref{eq:third_constraint_C_types} will be satisfied if the following holds:
\begin{equation*}
    (T-2)\cdot (V^\theta(\x^*, j^{*\theta}(\x^*)) - \max_{j \notin \texttt{BR}(\Theta')}V^\theta(\x^*, j))\ge 1
\end{equation*}

Therefore, set $T^*$ such that:
\begin{small}
\begin{align*}
    T^* = \max\Bigg(\frac{1}{V^\theta(\x^*, j^{*\theta}(\x^*)) - \max_{j \notin \texttt{BR}(\Theta')}V^\theta(\x^*, j)} + 1, \frac{1}{V^\theta(\x^*, j^{*\theta}(\x^*)) - \max_{j \notin \texttt{BR}(\Theta')}V^\theta(\x^*, j)} + 2\Bigg),
\end{align*}
\end{small}
proves the theorem.%
\end{proof}

\subsection{Proof Details for Section~\ref{subsec:learning_random_games}}
\label{subappend:learning_random_games}
In this section, we will provide the details of the proofs of Theorem~\ref{thm:non_intersecting_BR_regions} and Corollary~\ref{cor:learning_is_effective_with_high_probability}.
To do so, we will define some additional notation, and in doing this we will follow the conventions of \citet{hugRandomConicalTessellations2016}.
A key step will be establishing Corollary~\ref{corollary:expected_size_of_schlafli_cones}, and the proof of this corollary is closely related, with some important modifications, to the proof of Corollary 4.2 in \citet{hugRandomConicalTessellations2016}.
Therefore, Corollary~\ref{corollary:expected_size_of_schlafli_cones} should be viewed as a corollary to Corollary 4.2 in \citet{hugRandomConicalTessellations2016}.
However, we allow for a generalized notion of the quermassintegral in our derivation, which is essential to prove Corollary~\ref{corollary:expected_size_of_schlafli_cones}, making our result a strict generalization of Corollary 4.2 in \citet{hugRandomConicalTessellations2016}.

Consider the $m$-dimensional space, $\mathbb{R}^m$ and the $m$-dimensional unit sphere denoted by $\mathbb{S}^{m-1}$.
Let $\mathcal{D}^m$ be the set of closed, convex cones in $\mathbb{R}^m$, and let $\mathcal{PC}^m$ be the set of polyhedral cones in $\mathbb{R}^m$ where $\mathcal{PC}^m \subset \mathcal{D}^m$.
For all $D \in \mathcal{D}^d$, the set $K = D \cap \mathbb{S}^{m-1}$ is a convex body in $\mathbb{S}^{m-1}$.
We denote the set of convex bodies in $\mathbb{S}^{m-1}$ as $\mathcal{K}_\mathcal{S}$. We also denote by $G(m,k)$ the Grassmanian of $k$-dimensional linear subspaces of $\mathbb{R}^m$.


\begin{definition}[Generalized Spherical $(j,\phi)$-Quermassintegral]
    Let $\phi^*(\cdot)$ be a normalized measure over $G(m,m-j)$ where $j \in \{0,...,m\}$. Then the \emph{generalized spherical $j$-quermassintegral} for $K \in \mathcal{K}_\mathcal{S}$ is:
    \begin{equation}
        U_{j,\phi^*} = \frac{1}{2} \int_{G(m,m-j)} \chi(K \cap L) \phi(dL).
    \end{equation}
    Where $\chi(\cdot)$ is the Euler characteristic.
\end{definition}

Note that if $\phi^*=\nu_{m-j}$, the normalized Haar measure on $G(m, m-j)$, then the generalized spherical $(j,\nu_{m-j})$-quermassintegral is identical to the standard spherical $j$-quermassintegral.
If $K$ is a convex body in $\mathbb{S}^{m-1}$ that is not a great sub-sphere, $\chi(K \cap L) = \mathbbm{1}(K \cap L \ne \emptyset)$ for almost all $L \in G(m,m-j)$, see p. 40 of \citet{schneiderStochasticIntegralGeometry2008}.
Therefore if $\phi$ is an absolutely continuous measure with respect to $\nu_{m-j}$, $2U_{j,\phi}(K)$ is the probability measure of a random $(m-j)$-dimensional linear subspace, sampled according to $\phi$, intersecting with the set $K$.

Consider the set of hyperplanes $H_1, H_2, ..., H_n \in G(m,m-1)$. 
We will follow \citet{hugRandomConicalTessellations2016} and say that they are in \emph{general position} if any $k \le m$ of them have an intersection of dimension $m-k$.


If the hyperlanes $H_1, H_2, ..., H_n \in G(m,m-1)$ are in general position, they induce a conical tessalation \citep{hugRandomConicalTessellations2016} of $\mathbb{R}^m$ into $m$-dimensional polyhedral cones, denoted by $\mathcal{T}$.
We will write $\mathcal{F}(H_1, H_2, ..., H_n)$ to denote the set of polyhedral cones in $\mathcal{T}$.
For $C\in\mathcal{F}(H_1, H_2, ..., H_n)$, the convex sets $C \cap \mathbb{S}^{m-1}$ form a spherical tessalation of $\mathbb{S}^{m-1}$.
If we denote arbitrarily by $H^-$ one of the two half-spaces bounded by the hyperplane, then each of the $m$-dimensional cones defined by the tessalation $\mathcal{T}$ induced by $H_1, H_2, ..., H_n$ are of the form:
\begin{equation}
    \bigcap_{i=0}^n \epsilon_i H_i^-, \epsilon_i = \pm 1
\end{equation}
for some set of $\epsilon_i$.
We call these cones the \emph{Schl{\"a}fli cones}, following \citet{hugRandomConicalTessellations2016}.

\begin{lemma}[Lemma 8.2.1 in \citet{schneiderStochasticIntegralGeometry2008}]
    For a set of hyperplanes $H_1, H_2, ..., H_n \in G(m,m-1)$ in general position, there are $C(n,m)$ distinct Schl{\"a}fli Cones where
    \begin{equation}
        C(n,m) = 2 \sum_{i=0}^{m-1} {n-1 \choose i}.
    \end{equation}
\end{lemma}

We will denote by $\mathcal{B}(T)$ the $\sigma$-algebra of Borel sets for a given topological space $T$, typically either $\mathbb{S}^{m-1}$ or $G(m,m-1)$.
Let $\phi'$ be a probability measure over $\mathcal{B}(\mathbb{S}^{m-1})$ that is even and assigns measure zero to each $(m-2)$-dimensional great sub-sphere of $\mathbb{S}^{m-1}$.
Then let $\x_1, \x_2, ..., \x_n$ be random points on $\mathbb{S}^{m-1}$ distributed according to $\phi'$.
The corresponding random hyperplanes, $\{\boldsymbol{y}\in\mathbb{R}^m \mid \x_i\cdot \boldsymbol{y} = 0\}_{i\in [n]}$, will be in general position with probability $1$.
We will denote these random hyperplanes as $\mathcal{H}_1, \mathcal{H}_2, ..., \mathcal{H}_n$, and the induced probability measure over these hyperplanes, we will denote $\phi$.

\begin{definition}[Random Schl{\"a}fli Cones]
    Let $\phi$ be a probability measure as indicated above. Let $n\in \mathbb{N}$ and let $\mathcal{H}_1, \mathcal{H}_2, ..., \mathcal{H}_n \in G(m, m-1)$ be independent random hyperplanes distributed according to $\phi$. The random \emph{$(\phi, n)$-Schl{\"a}fli cone $\mathcal{S}_n$} is the cone chosen uniformly at random from the polyhedral cones induced by $\mathcal{H}_1, \mathcal{H}_2, ..., \mathcal{H}_n$.
\end{definition}

With a slight abuse of notation, denote by $\mathbf{H}_n = H_1, H_2, ..., H_n$.
We follow \citet{hugRandomConicalTessellations2016} (see p. 406) and formally define $S_n$ as the polyhedral cone with distribution given by:
    \begin{equation}
        \label{eq:random_cone}
        \begin{aligned}
            \mathbb{P}(\mathcal{S}_n \in B) = \int_{G(m,m-1)^n}\frac{1}{C(n,m)}  \cdot \sum_{C \in \mathcal{F}(\mathbf{H}_n)} \mathbbm{1}_B(C) \phi^n(d(\mathbf{H}_n))
        \end{aligned}  
    \end{equation} 

Now, let $L \in G(m, k)$ be in general position with respect to $H_1, H_2, ..., H_n \in G(m,m-1)$.
Then $L \cap H_1, ..., L\cap H_n$ are $(k-1)$-dimensional subspaces of $L$.
Moreover, they are in general position in $L$.
Therefore, there are $C(n,k)$ $k$-dimensional Schl{\"a}fli Cones cones in the tessalation $\mathcal{T}_L$, i.e. the tessalation induced in $L$ by $H_1, H_2, ..., H_n$.

Again, if we let $\mathcal{H}_1, \mathcal{H}_2, ..., \mathcal{H}_n \in G(m,m-1)$ be distributed according to $\phi$, and let $\mathcal{L} \in G(m,k)$ be distributed according to $\phi^*$ (satisfying the same conditions as $\phi$), then we choose a random $k$-dimensional cone in $\mathcal{T}_\mathcal{L}$.
Note that $\mathcal{H}_1, \mathcal{H}_2, ..., \mathcal{H}_n$ and $\mathcal{L}$ are in general position with probability $1$.
This random cone, denoted $\mathcal{C}_n^k$ is precisely the polyhedral cone given by the following:
\begin{small}
    \begin{equation}
        \label{eq:random_intersect_cone}
        \begin{aligned}    
            P(\mathcal{C}_n^k \in B) = \int_{G(m,m-1)} \int_{G(m,k)}\frac{1}{C(n,k)} 
             \cdot \sum_{C \in \mathcal{F}_d(\mathbf{H}_n)}  1(C\cap L \ne \{0\}) \mathbbm{1}_B(C) \phi^*(dL) \phi^n(d(\mathbf{H}_n))
        \end{aligned}
    \end{equation}
\end{small}

\begin{corollary}[Corollary to Corollary 4.2 in \citet{hugRandomConicalTessellations2016}]
    \label{corollary:expected_size_of_schlafli_cones}
    Let $\phi$ be a probability measure as indicated above, and let $n\in \mathbb{N}$ and $\mathcal{H}_1, \mathcal{H}_2, ..., \mathcal{H}_n \in G(m, m-1)$ be independent random hyperplanes distributed according to $\phi$.
    Then the probability that a random hyperplane, $\mathcal{L}$, distributed according to $\phi^*$ intersects with a random \emph{$(\phi, n)$-Schl{\"a}fli cone $\mathcal{S}_n$} is:
    \begin{equation}
        E(U_{k, \phi^*}(\mathcal{S}_n)) = \frac{C(n,m-k)}{2C(n,m)}
    \end{equation}
\end{corollary}

\begin{proof}{Proof.}
    Consider the inner integral to equation~\ref{eq:random_intersect_cone}:
    \begin{small}
    \begin{align*}
        &\int_{G(m,k)}\sum_{C \in \mathcal{F}_d(\mathbf{H}_n)}\mathbbm{1}(C\cap L \ne \{0\})1_B(C) \phi^*(dL) \\
        & \quad = \sum_{C \in \mathcal{F}_d(\mathbf{H}_n)}\mathbbm{1}_B(C)\int_{G(m,k)}1(C\cap L \ne \{0\})\phi^*(dL) \\
        & \quad = \sum_{C \in \mathcal{F}_d(\mathbf{H}_n)}\mathbbm{1}_B(C) 2U_{m-k, \phi^*}(C)
    \end{align*}
    \end{small}
    Therefore,
    \begin{small}
    \begin{align}
        \label{eq:proof_prob_equation}
        \begin{aligned}
            P(\mathcal{C}_n^k \in B) &= \frac{2}{C(n,k)}\int_{G(m,m-1)}\sum_{C \in \mathcal{F}_d(\mathbf{H}_n)}  \cdot \mathbbm{1}_B(C) U_{m-k, \phi^*}(C)\phi^n(d(\mathbf{H}_n))  
        \end{aligned}
    \end{align}
    \end{small}

    From equations~\ref{eq:proof_prob_equation} and \ref{eq:random_cone}, for a non-negative and measurable function $g$ on $\mathcal{PC}^m$, the following holds:
    \begin{equation}
        \mathbb{E}[g(\mathcal{C}_n^k)] = \frac{2 C(n,m)}{C(n,k)}\mathbb{E}[g(\mathcal{S}_n)U_{m-k, \phi^*}(\mathcal{S}_n)]
    \end{equation}
    If we choose $g$ to be the identity function, then
    \begin{equation}
        \mathbb{E}[U_{m-k, \phi^*}(\mathcal{S}_n)] = \frac{C(n,k)}{2 C(n,m)}
    \end{equation}
    as desired.%
\end{proof}

We will use the following corollary and lemma in our proof of Theorem~\ref{thm:non_intersecting_BR_regions}.

\begin{corollary}
    \label{cor:prob_of_intersection}
    For a set of random half-spaces $\{\mathcal{H}_1^{+}, \mathcal{H}_2^{+}, ..., \mathcal{H}_n^{+}\}$ distributed according to $\phi$ that satisfies Assumption~\ref{assum:generic_condition}, the probability that a hyperplane sampled from $\phi$ intersects the intersection of the half-spaces is:
    \begin{equation}
        \mathbb{E}\left[\chi\left(\bigcap_{i \in [n]}\mathcal{H}_i^+ \cap \mathcal{H}'\right)\right] = \frac{\sum_{i=0}^{m-2}{n-1 \choose i}}{\sum_{i=0}^{m-1}{n-1 \choose i}}
    \end{equation}
\end{corollary}

\begin{proof}{Proof.}
    The corollary is a direct consequence of Corollary~\ref{corollary:expected_size_of_schlafli_cones} with $\phi^* =\phi$.%
\end{proof}

\begin{lemma}\label{lemma:upper_bound_sum_choose}
    Let $N,k \in \mathbb{N}$, with $k < \frac{N+1}{2}$. Then
    \begin{equation}
        \sum_{i=0}^{k} {N \choose i} \le {N \choose k}\frac{N - k + 1}{N - 2k + 1} 
    \end{equation}
\end{lemma}

\begin{proof}{Proof.}
    Consider the sum:
    \begin{small}
    \begin{align}
        \frac{\sum_{i=0}^{k} {N \choose i}}{{N \choose k}} &= 1 + \frac{k}{N-k+1} + \frac{k(k-1)}{(N-k+1)(N-k+2)}  + \cdots \\
         &\le 1 + \frac{k}{N-k+1} + \left(\frac{k}{N-k+1}\right)^2 + \cdots \\
        & = \frac{N - k + 1}{N - 2k + 1} 
    \end{align}
    \end{small}
    Note that the final equality is due to $k < \frac{N+1}{2}$ implying that $\frac{k}{N-k+1} < 1$.%
\end{proof}

We are now ready to prove the main technical result, Theorem~\ref{thm:non_intersecting_BR_regions}.

\begin{proof}{Proof of Theorem~\ref{thm:non_intersecting_BR_regions}.}
    Note that by Boole's Inequality
    \begin{small}
    \begin{align*}
        P\left(\bigcup_{i \in [n]}\left(\bigcap_{\theta \in\Theta} \texttt{BR}_\theta(i)\right) \ne \emptyset\right) 
        & \le \sum_{i=1}^n P\left(\bigcap_{\theta \in\Theta} \texttt{BR}_\theta(i) \ne \emptyset\right).
    \end{align*} 
    \end{small}

    Let us consider then the probability $P\left(\bigcap_{\theta \in\Theta} \texttt{BR}_\theta(i) = \emptyset\right)$.
    It is easy to see that $\bigcap_{\theta \in\Theta} \texttt{BR}_\theta(i) = \bigcap_{\substack{ \theta \in \Theta \\ j \in [n]\setminus\{i\}}}H_{i,j}^{\theta,+} \cap \Delta^m$.
    Let $\Theta' \subseteq \Theta$ and let $j$ be such that the distribution $\phi$ of $\mathcal{H}_{i,j}^{\theta'}$ is even and assigns a measure zero to linear subspaces of $\mathbb{R}^m$.
    Then, $\bigcap_{\theta \in\Theta} \texttt{BR}_\theta(i) \subseteq \bigcap_{\theta', \in \Theta'}H_{i,j}^{\theta',+}$.
    Moreover, given Assumption~\ref{assum:generic_condition} and Corollary~\ref{cor:prob_of_intersection}, the expected probability that another random plane, $\mathcal{H}$, distributed according to $\phi$ will intersect $\bigcap_{\theta' \in \Theta'}H_{i,j}^{\theta',+}$ is given by:
    \begin{equation*}
        \mathbb{E}\left[\chi\left(\bigcap_{\theta' \in \Theta'}\mathcal{H}_{i,j}^{\theta',+} \cap \mathcal{H}\right)\right] =\frac{\sum_{i=0}^{m-2}{|\Theta'|-1 \choose i}}{\sum_{i=0}^{m-1}{|\Theta'|-1 \choose i}} \\
    \end{equation*}
    For notational simplicity, we will denote $k = |\Theta'|$. Then, simplifying
    \begin{align*}
        \mathbb{E}\left[\chi\left(\bigcap_{\theta' \in \Theta'}\mathcal{H}_{i,j}^{\theta',+} \cap \mathcal{H}\right)\right] =\frac{\sum_{i=0}^{m-2}{k-1 \choose i}}{{k-1 \choose m-1} +\sum_{i=0}^{m-2}{k-1 \choose i}}
    \end{align*}
    Note that by Lemma~\ref{lemma:upper_bound_sum_choose} and given that $\frac{k}{2} > m-2$ we can define the function $p(k)$ as an upper bound on the above expectation as follows:
    \begin{align}
        p(k) = \frac{{k-1 \choose m-2}\frac{k-1 - (m-3)}{k - 1 - (2m - 5)}}{{k-1 \choose m-1} + {k-1 \choose m-2}\frac{k-1 - (m-3)}{k - 1 - (2m - 5)}}\ge \frac{\sum_{i=0}^{m-2}{k-1 \choose i}}{{k-1 \choose m-1} +\sum_{i=0}^{m-2}{k-1 \choose i}} = \mathbb{E}\left[\chi\left(\bigcap_{\theta' \in \Theta'}\mathcal{H}_{i,j}^{\theta',+} \cap \mathcal{H}\right)\right]
    \end{align}
    It is trivial to see that $p(k)$ is decreasing in $k$ for $k>0$.
    Therefore, for a $\underline{k}$ such that $p(\underline{k}) \le \overline{p}$, for any $k \ge \underline{k}$, $\overline{p} \ge \mathbb{E}\left[2U_{1,\phi}\left(\bigcap_{\theta \in \Theta'}H_{i,j}^{\theta,+}\right)\right]$.
    For a given $p$ we can solve for the $k$ such that $p(k) = p$ as follows:
    \begin{equation*}
        p{k-1 \choose m-1} = (1-p){k-1 \choose m-2}\frac{k-1 - (m-3)}{k - 1 - (2m - 5)}.
    \end{equation*}
    By re-arranging and canceling:
    
    \begin{align*}
        (k - m + 1)(k - 2m + 4) = \frac{1-p}{p}(m-1)(k-m+2).
    \end{align*}
    Let us denote $q = \frac{1-p}{p}$.
    Then, we can rewrite this as a quadratic equation:
    \begin{align*}
        0 = k^2 + k(5 - 3m - q(m-1)) + (2+q)m^2  -(6+3q)m + (4+2q)
    \end{align*}
    This can be solved by the quadratic formula (choosing the positive root) as follows:
    \begin{small}
    \begin{align*}
        k &=\frac{1}{2} \bigg(-((q+1)(1-m)+4-2m) +\sqrt{m^2(q^2+2q+1) + m(-2q^2-4q-6) + q^2+2q+9} \bigg) \\
        &= \frac{1}{2} \bigg(((q+1)(m+1)-4+2m) +\sqrt{(m-1)^2(q+1)^2 - 4m+8}\bigg)
    \end{align*}
    \end{small}
    Note that $m\ge2$ implies $ - 4m+8\le 0$. Therefore if we choose $\underline{k}$ such that
    \begin{align*}
        \underline{k} &\ge \frac{1}{2}\bigg(((q+1)(m+1)-4+2m) +  \sqrt{(m-1)^2(q+1)^2} \bigg) \\
        &\ge \frac{1}{2} \bigg(((q+1)(m+1)-4+2m) +   \sqrt{(m-1)^2(q+1)^2 - 4m+8}\bigg)
    \end{align*}
    Therefore, for $\underline{q} = \frac{1-\overline{p}}{\overline{p}}$ set $\underline{k} = m(\underline{q}+2) = m\left(\frac{1}{\overline{p}} + 1\right)$, then for all $k \ge \underline{k}$: 
    \begin{equation}
        \label{eq:expectation_bound}
        \overline{p} \ge p(k) \ge \mathbb{E}\left[\chi\left(\bigcap_{\theta' \in \Theta'}\mathcal{H}_{i,j}^{\theta',+} \cap \mathcal{H}\right)\right].
    \end{equation}

    
    Equation~\ref{eq:expectation_bound} is a bound on the \emph{expected} probability that a random hyperplane from $G(m,m-1)$ drawn according to $\phi$ intersects the set $\bigcap_{\theta \in \Theta'}H_{i,j}^{\theta,+}$.
    The probability can be bounded directly by using the Markov inequality:
    \begin{equation}
        P\left(\chi\left(\bigcap_{\theta' \in \Theta'}\mathcal{H}_{i,j}^{\theta',+} \cap \mathcal{H}\right) > d \underline{p} \right) \le \frac{1}{d}.
    \end{equation}

    Now consider the $(|\Theta| - |\Theta'|)$-hyperplanes generated by the set of types $\Theta \setminus \Theta'$.
    Consider a random hyperplane in this set $H_{i,j}^{\theta}$ for some $\theta \in \Theta \setminus \Theta'$.
    With probability greater than $1 - \frac{1}{d}$, the chance that $H_{i,j}^{\theta}$ intersects $\bigcap_{\theta' \in \Theta'}H_{i,j}^{\theta',+}$ is less than $d\overline{p}$.
    If $H_{i,j}^{\theta}$ does not intersect $\bigcap_{\theta' \in \Theta'}H_{i,j}^{\theta',+}$, then there is a $\frac{1}{2}$ chance that $\bigcap_{\theta' \in K}H_{i,j}^{\theta',+} \cap H_{i,j}^{\theta,+} = \emptyset$, by the assumption that $\phi$ is even.

    Denote by $J$ the condition such that $\chi\left(\bigcap_{\theta' \in \Theta'}\mathcal{H}_{i,j}^{\theta',+} \cap \mathcal{H}\right) \le d \underline{p}$
    \begin{small}
    \begin{align}
        &P\left(\bigcap_{\theta \in\Theta} \texttt{BR}_\theta(i) \ne \emptyset\right) \le P\left(\bigcap_{\theta' \in \Theta'}\mathcal{H}_{i,j}^{\theta',+} \cap \bigcap_{\theta \in \Theta \setminus \Theta}\mathcal{H}_{i,j}^{\theta,+} \ne \emptyset\right)\\
        &\begin{aligned}
            \label{eq:total_probability_bound}
            =P&\left(\bigcap_{\theta' \in \Theta'}\mathcal{H}_{i,j}^{\theta',+} \cap \bigcap_{\theta \in \Theta \setminus \Theta}\mathcal{H}{i,j}^{\theta,+} \ne \emptyset \mid J\right)P\left(J \right) 
             +P\left(\bigcap_{\theta' \in \Theta'}\mathcal{H}_{i,j}^{\theta',+} \cap \bigcap_{\theta \in \Theta \setminus \Theta}\mathcal{H}{i,j}^{\theta,+} \ne \emptyset \mid \neg J\right)P\left(\neg J \right)
        \end{aligned} \\
        &\le P\left(\bigcap_{\theta' \in \Theta'}\mathcal{H}_{i,j}^{\theta',+} \cap \bigcap_{\theta \in \Theta \setminus \Theta}\mathcal{H}_{i,j}^{\theta,+} \ne \emptyset \mid J\right) + \frac{1}{d} \\
        &=P\left(\bigcap_{\theta \in \Theta \setminus \Theta'} \bigcap_{\theta \in \Theta'}\mathcal{H}_{i,j}^{\theta',+} \cap \mathcal{H}_{i,j}^{\theta,+} \ne \emptyset \mid J\right) + \frac{1}{d} \\
        \label{eq:intersection_of_empty_sets}
        &\le P\left(\bigcap_{\theta \in \Theta \setminus \Theta'} \left(\bigcap_{\theta' \in \Theta'}\mathcal{H}_{i,j}^{\theta',+} \cap \mathcal{H}_{i,j}^{\theta,+} \ne \emptyset\right) \mid J\right) + \frac{1}{d} \\
        \label{eq:independence_of_H}
        &= \prod_{\theta \in \Theta \setminus \Theta'}P\left(\bigcap_{\theta' \in \Theta'}\mathcal{H}_{i,j}^{\theta',+} \cap \mathcal{H}_{i,j}^{\theta,+} \ne \emptyset \mid J\right) + \frac{1}{d}.
    \end{align}
    \end{small}
    Note that equation~\ref{eq:total_probability_bound} is due to the law of total probability, equation~\ref{eq:intersection_of_empty_sets} is true since if the member of an intersection of sets is empty then the intersection will also be empty, and equation~\ref{eq:independence_of_H} is due to the independence across types in Assumption~\ref{assum:generic_condition}.
    Consider the first term in equation~\ref{eq:independence_of_H}.
    Denote by $K = \chi\left(\bigcap_{\theta' \in \Theta'}H_{i,j}^{\theta',+} \cap H_{i,j}^{\theta}\right)$
    By the previous discussion, again using the law of total probability and noting that the Euler characteristic $K = \chi\left(\bigcap_{\theta' \in \Theta'}H_{i,j}^{\theta',+} \cap H_{i,j}^{\theta}\right)$ will be $1$ only if $\bigcap_{\theta' \in \Theta'}H_{i,j}^{\theta',+} \cap H_{i,j}^{\theta,+} \ne \emptyset$:
    \begin{small}
    \begin{align*}
        &\prod_{\theta \in \Theta \setminus \Theta'}P\left(\bigcap_{\theta' \in \Theta'}\mathcal{H}_{i,j}^{\theta',+} \cap \mathcal{H}_{i,j}^{\theta,+} \ne \emptyset \mid J\right) 
        \\ &  =\prod_{\theta \in \Theta \setminus \Theta'} \Biggl(P\left(K\mid J\right) 
         + P\left(\neg K\mid J\right) P\left(\bigcap_{\theta' \in \Theta'}\mathcal{H}_{i,j}^{\theta',+} \cap \mathcal{H}_{i,j}^{\theta,+} \ne \emptyset \mid J, \neg K\right)\Biggl) \\
        &  \le \prod_{\theta \in \Theta \setminus \Theta'}\left(d\overline{p} + (1-d\overline{p})\left(\frac{1}{2}\right)\right) \\
        & = \left(\frac{1+d\overline{p}}{2}\right)^{|\Theta| - |\Theta'|}.
    \end{align*}
    \end{small}
    Therefore
    \begin{align*}
        P\left(\bigcap_{\theta \in\Theta} \texttt{BR}_\theta'(i) \ne \emptyset\right) \le \left(\frac{1+d\overline{p}}{2}\right)^{|\Theta| - |\Theta'|} + \frac{1}{d}.
    \end{align*}
    Set $d = \frac{1}{2\overline{p}}$ and choose a set $\Theta'$ such that $|\Theta'| = |\Theta| - \left\lceil\log_{\frac{3}{4}}\frac{1}{|\Theta|}\right\rceil$. Given that $\overline{p} = \frac{m}{|\Theta'| + m}$ then
    \begin{small}
    \begin{align*}
        \left(\frac{1+d\overline{p}}{2}\right)^{|\Theta| - |\Theta'|} + \frac{1}{d} &= \left(\frac{3}{4}\right)^{\left\lceil\log_{\frac{3}{4}}\frac{1}{|\Theta|}\right\rceil} + 2\frac{m}{|\Theta| - \left\lceil\log_{\frac{3}{4}}\frac{1}{|\Theta|}\right\rceil + m} \\
        & \le \left(\frac{3}{4}\right)^{\log_{\frac{3}{4}}\frac{1}{|\Theta|}} + 2\frac{m}{|\Theta| - \left\lceil\log_{\frac{3}{4}}\frac{1}{|\Theta|}\right\rceil + m} \\
        &= \frac{1}{|\Theta|} + 2\frac{m}{|\Theta| - \left\lceil\log_{\frac{3}{4}}\frac{1}{|\Theta|}\right\rceil + m}
    \end{align*}
    \end{small}

    Therefore
    \begin{align*}
        P\Bigg(&\bigcup_{i \in [n]}\left(\bigcap_{\theta \in\Theta} \texttt{BR}_\theta(i)\right) \ne \emptyset \Bigg) \le \sum_{i=1}^n P\left(\bigcap_{\theta \in\Theta} \texttt{BR}_\theta'(i) \ne \emptyset\right) \\
        &\le n\left(\frac{1}{|\Theta|} + 2\frac{m}{|\Theta| - \left\lceil\log_{\frac{3}{4}}\frac{1}{|\Theta|}\right\rceil + m}\right) = O\left(\frac{1}{|\Theta|}\right)
    \end{align*}
    for a fixed $n$ and $m$.%
\end{proof}

Using Theorem~\ref{thm:non_intersecting_BR_regions}, we can prove Corollary~\ref{cor:learning_is_effective_with_high_probability}.

\begin{proof}{Proof of Corollary~\ref{cor:learning_is_effective_with_high_probability}.}
    First, note that by Assumption~\ref{assum:generic_condition} the set of hyperplanes $H_{i,j}^\theta$ for all $i,j \in [n]$ and $\theta \in \Theta$ are in general position.
    Additionally, if the hyperplanes that define the boundary of the simplex $\Delta^m$, $\{\x \mid x_i = 0\}_{i\in [m]}$, are added to the set, they are still in general position with probability one.
    There are two cases that must be considered either the $\texttt{BSE}(\Theta)$, denoted $\x^*$, is at the vertex of the simplex $\Delta^m$, i.e. there exists some $i \in [m]$ such that $x^*_i = 1$, or $\x^*$ is not at the vertex of the simplex.
    We will first consider the case where $\x^*$ is not at the vertex of the simplex.
    In this case, $\x^*$ must be at the intersection of exactly $m$ distinct hyperplanes, where at least one of those hyperplanes is $H_{i,j}^{\theta^*}$ for some $\theta^* \in \Theta$ and $i,j \in [n]$.
    This is due to $\x^*$ not being at the vertex of the simplex and the leader's objective function being linear.

    By Assumption~\ref{assum:generic_condition} and Theorem~\ref{thm:non_intersecting_BR_regions}, with probability greater than $1 - O\left(\frac{1}{|\Theta|}\right)$ there exists a $\Theta' \subset \Theta$ such that $\texttt{BR}(\Theta') \cap \texttt{BR}(\{\theta^*\}) = \emptyset$.
    Note that this implies that $\theta^* \notin \Theta'$.
    Therefore, $\texttt{BSE}(\Theta') \ne \x^*$ since the hyperplane $H_{i,j}^{\theta^*}$ is not in the set of indifference curves for $\Theta'$.

    Now, we consider the case where $\x^*$ is a vertex point of $\Delta^m$.
    We will denote by $\x^i$ the strategy $x_{i} = 1$ and $x_{j} = 0$ for all $j \in [m] \setminus \{i'\}$.
    Therefore, there exists an $i^* \in [m]$ such that $\x^* = \x^{i^*}$.
    First, note that there are $m$ vertices, and each vertex is the intersection of exactly $m$ hyperplanes.
    Moreover, with probability one, each follower type $\theta \in \Theta$ will have exactly one best response, due to the set of hyperplanes $H_{i,j}^\theta$ for all $i,j \in [n]$ and $\{\x \mid x_i = 0\}_{i\in [m]}$ being in general position.
    In this case, by the assumption that leader action $i^* \in [m]$ is not a dominant strategy, there exists some follower action $j' \in [n]$ and some other leader action $i' \in [m]$ such that $R_{i^*,j'} \le R_{i',j'}$.
    Therefore, if there is some subset of types $\Theta' \in \Theta$ such that $\texttt{BR}(\Theta') = \{j'\}$, then $\texttt{BSE}(\Theta') \ne \x^{i^*}$ since the strategy $(1 - \epsilon)\x^{i^*} + \epsilon\x^{j'}$ is a strict improvement for a sufficiently small $\epsilon$, given that the indifference curves and boundaries of the simplex are in general position and $\boldsymbol{\mu}$ has full support.

    Then it suffices to show that the probability that there exists a vertex of the simplex such that not all follower actions are a best response for at least one follower type is sufficiently small.
    By the symmetry of Assumption~\ref{assum:generic_condition} for a random leader strategy $\x$, the probability that any given $j \in [n]$ is a best response for type $\theta$ is uniform over $[n]$.
    Therefore, define $E_{i,j}^\theta$ to be the event where follower action $j \in [n]$ is a best response to leader strategy $\x^i$ for $\theta \in \Theta$ and $i \in [m]$.
    Then, we must bound the probability that for at least one $i$ and $j$ $E_{i,j}^\theta$ does not happen for all $\theta \in \Theta$:
    \begin{align*}
        p &= P\left(\bigcup_{i \in [m]}\bigcup_{j \in [n]}\bigcap_{\theta \in \Theta} \neg E_{i,j}^\theta\right) 
         \le mnP\left(\bigcap_{\theta \in \Theta} \neg E_{i,j}^\theta\right) 
        = mn \left(\frac{n-1}{n}\right)^{|\Theta|} 
         \le mn e^{-\frac{|\Theta|}{n}} = O\left(e^{-|\Theta|}\right)
    \end{align*}
    Since $O\left(\frac{1}{|\Theta|}\right) \ge O\left(e^{-|\Theta|}\right)$, the probability that Assumption~\ref{assumption:learnable_subgroup} is not satisfied is $O\left(\frac{1}{|\Theta|}\right)$.%
\end{proof}

\subsection{Additional Simulation Results}
\label{subsection:effective_learing_additional_simulation_results}

Table~\ref{table:assumption_justify_2} shows the proportion of random games that permit effective learning when all game parameters for $R$ and $C$ are drawn IID from the standard normal distribution.
The insights are identical to Table~\ref{table:assumption_justify}.

\begin{table}[tbh]
    \begin{minipage}{0.24\linewidth}\resizebox{\textwidth}{!}{
    \begin{tabular}{|c|c|c|c|} 
     \hline
       $|\Theta|=2$ & $n=5$  & $n=10$  & $n=15$\\ 
     \hline
     $m=5$  & 48/100  & 51/100& 72/100\\
     \hline
     $m=10$  & 27/100  & 50/100& 62/100\\
     \hline
     $m = 15$  & 25/100   & 48/100& 50/100 \\
     \hline
    \end{tabular}}
    \end{minipage}
    \begin{minipage}{0.24\linewidth}\resizebox{\textwidth}{!}{
    \begin{tabular}{|c|c|c|c|} 
     \hline
       $|\Theta|=3$ & $n=5$  & $n=10$  & $n=15$\\ 
     \hline
     $m=5$  & 66/100  & 83/100& 93/100\\
     \hline
     $m=10$  & 48/100  & 75/100 & 87/100\\
     \hline
     $m = 15$  & 40/100   & 70/100 & 84/100\\
     \hline
    \end{tabular}}
    \end{minipage}
    \begin{minipage}{0.24\linewidth}\resizebox{\textwidth}{!}{
        \begin{tabular}{|c|c|c|c|} 
     \hline
       $|\Theta|=4$ & $n=5$  & $n=10$  & $n=15$\\ 
     \hline
     $m=5$  & 70/100  & 90/100& 100/100\\
     \hline
     $m=10$  & 62/100  & 86/100 & 93/100\\
     \hline
     $m = 15$  & 48/100   & 86/100 & 92/100\\
     \hline
    \end{tabular}}
    \end{minipage}
    \begin{minipage}{0.24\linewidth}\resizebox{\textwidth}{!}{
        \begin{tabular}{|c|c|c|c|} 
     \hline
       $|\Theta|=5$ & $n=5$  & $n=10$  & $n=15$\\ 
     \hline
     $m=5$  & 74/100  & 94/100& 99/100\\
     \hline
     $m=10$  & 59/100  & 87/100& 96/100\\
     \hline
     $m = 15$  & 60/100   &  87/100 & 95/100\\
     \hline
    \end{tabular}}
    \end{minipage}
    \caption{Probability of Assumption~\ref{assumption:learnable_subgroup} Being Satisfied Under $100$ Normal Random Bayesian Stackelberg Game Instances. Probability of assumption \ref{assumption:learnable_subgroup} being satisfied under $100$ random bayesian Stackelberg game instances. To generate each instance, we sample every game parameter  $R_{i,j}$ and $C^\theta_{i,j}$  $\forall i \in [m], j\in [n], \theta \in \Theta$ independently from the standard normal distribution. The prior distribution $\bmu$ is uniform over $\Theta$. Each row represents the number of leader actions, and the column represents the number of follower actions.
    }
    \label{table:assumption_justify_2}
    \end{table}


\section{Omitted Details in the Proof of Theorem \ref{thm:infinte_T}}\label{sec:appendix_efficacy}

\begin{proof}{Proof of Theorem~\ref{thm:infinte_T}.}
The high-level proof idea is as follows. Given any $\RME$ $\{p^*_{\theta,j},\bx^*_{\theta, j}\}_{j\in [n], \theta \in \Theta}$, we want to construct a dynamic policy (not necessarily the optimal $\DSE$) that can ``simulate'' the given $\RME$. Unfortunately, this construction has two main challenges that we have to overcome.  
First, the $\RME$ offers a menu of randomized strategies $\{p^*_{\theta, j}, \bx^*_{\theta, j}\}_{j \in [n]}$ for each type $\theta$. 
Because of incentive compatibility, the follower reports type $\theta$ truthfully and the leader plays a randomly selected strategy. 
However, the dynamic policy plays a single strategy at each round. As a result, the dynamic policy must use the initial rounds to learn the follower's type $\theta$ and then simulate the corresponding randomized strategy on the $\RME$.  
Another challenge comes from the simulation of randomization over strategies with deterministically chosen strategies. 
An intuitive idea is to play $\bx^*_{\theta,j}$ for the number of rounds that is proportional to $p^*_{\theta,j}$. Unfortunately, this requirement cannot be exactly fulfilled  unless the number of rounds can be divided \textit{evenly} for every $j$ according to $p^*_{\theta,j}$; Otherwise, the follower's incentives may be distorted. 
An overview of the high-level idea is represented in the Figure~\ref{fig:sequece_of_leader_strategies}.
    \begin{figure}[tbh]
    \centering
    \includegraphics[width=0.7\linewidth]{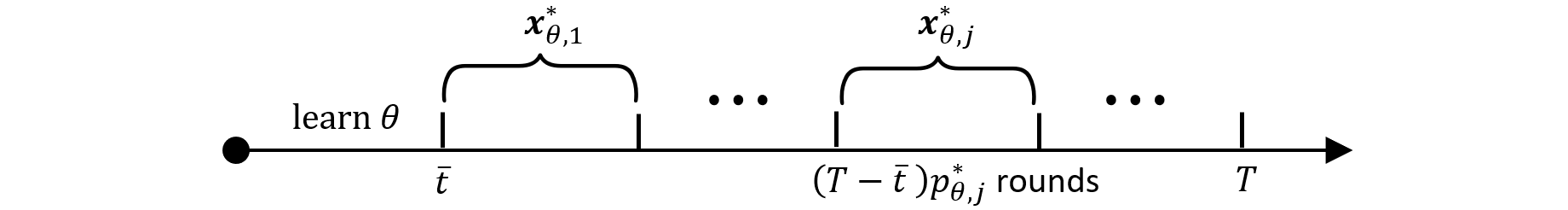}
    \caption{Illustration of the Sequence of Leader Strategies. The leader first plays a strategy for $\overline{t}$-rounds to elicit the followers type.
    Then the leader plays each strategy in the randomized set of strategies for approximately $(T - \overline{t})p^*_{\theta,j}$-rounds.}
    \label{fig:sequece_of_leader_strategies}
    \end{figure} 
    We show how to address these challenges in three steps. 

    \textbf{Step 1: Eliciting follower type $\theta$.} This phase takes $\bar{t}  = \lceil\log_n(|\Theta|)\rceil$ rounds. Construct function $H: \Theta \rightarrow [n]^{\bar{t}}$ such that $H(\theta)$ equals precisely the length-$\bar{t}$ \emph{bit sequence} of the $n$-nary representation of $\theta$ (interpreting $\theta$ as an integer equaling at most $|\Theta|$). By definition, each $\theta$ has a unique $H(\theta)$ value. The leader strategies in these first $\bar{t}$ rounds can be arbitrary. The main challenge is deferred to the second phase below, during which we will carefully design the strategies to incentivize each follower type $\theta$ to respond exactly with action sequence $H(\theta)$ in this first phase. 
    
    \textbf{Step 2: Constructing an approximately-optimal strictly-IC randomized menu.} There are two reasons why we need strict incentive compatibility when simulating the randomized menu. The first is to incentivize every follower $\theta$ to respond with $H(\theta)$ in the above Step 1.  To achieve this, we need to guarantee that, from $\theta$'s point of view, the strategy sequence in the remaining  $T-\bar{t}$ rounds for him is strictly better --- in fact at least additively $\bar{t}$ better --- than the strategy sequence for any other $\theta'$. The second reason is more intrinsic: we have to round all the probabilities in the randomized menu being played after $\bar{t}$ into multipliers of $1/(T - \bar{t})$, such that they can be \textit{precisely} realized by deterministic strategies. Because of this, we also need a strictly IC menu to make up for the incentive distortion during the rounding of probabilities. 
    In this step, we show that there always exists such a randomized menu that is $O\left(\sqrt{\frac{\log |\Theta|}{T\delta^2}}\right)$ optimal and $O\left(\sqrt{\frac{\log |\Theta|}{T}}\right)$-strictly IC (this is where the non-degenerate inducibility gap is needed). A high-level idea is to construct a new randomized menu that is both near-optimal and strictly IC by mixing the $\RME$ with the $\delta$-strictly IC menu.
    \begin{lemma}\label{lem:constructing}
    For Stackelberg games with inducibility gap $\delta$ and $T\geq \Omega(\bar{t}^2) $, there always exists a randomized menu that is $O\left(\sqrt{\frac{\log |\Theta|}{T\delta^2}}\right)$ optimal and $O\left(\sqrt{\frac{\log |\Theta|}{T}}\right)$-strictly IC. 
    \end{lemma}
    



    \begin{proof}{Proof of Lemma~\ref{lem:constructing}}
    We prove the lemma by explicitly constructing a randomized menu $\langle \bm{p}, \bx \rangle$ and show that $\langle \bm{p}, \bx \rangle$ is $O\left(\sqrt{\frac{\log |\Theta|}{T}}\right)$-strictly IC and achieves $\left(1 - \sqrt{\frac{\log |\Theta|}{T\delta^2}}\right)U^{\RME}$. Specifically, denote the $\RME$ as $\langle \bm{p}^*, \bx^* \rangle$ and the $\delta$-strictly IC randomized menu as $\langle \bm{p}^\delta, \bx^\delta \rangle$,  we construct the new randomized menu $\langle \bm{p}, \bx \rangle$ for all $\theta \in \Theta$ and $j \in [n]$ as follows 
    \begin{small}
    \begin{align*}
    &p_{\theta,j} = \left(1-\sqrt{\frac{\log |\Theta|}{T\delta^2}}\right) p^*_{\theta,j} + \left(\sqrt{\frac{\log |\Theta|}{T\delta^2}}\right) p^\delta_{\theta, j} \\ 
    &\bx_{\theta,j} = \frac{\left(1-\sqrt{\frac{\log |\Theta|}{T\delta^2}}\right) p^*_{\theta,j}}{p_{\theta,j}} \bx^*_{\theta, j} + \frac{\left(\sqrt{\frac{\log |\Theta|}{T\delta^2}}\right) p^\delta_{\theta, j}}{p_{\theta,j}} \bx^\delta_{\theta, j}
    \end{align*}
    \end{small}
    For every follower type $\theta \in \Theta$ and $\theta' \in \Theta$ such that $\theta' \ne \theta$, by definition
    \begin{small}
    \begin{equation}
    \label{eq:IC}
    \begin{aligned}
    &\sum_j p^\delta_{\theta,j} V^\theta(\bx^\delta_{\theta,j}, j)  \ge \sum_j p^\delta_{\theta',j} \max_{j'} V^\theta(\bx^\delta_{\theta',j}, j') + \delta, \\
    &\sum_j p^*_{\theta,j} V^\theta(\bx^*_{\theta,j}, j)  \ge \sum_j p^*_{\theta',j} \max_{j'} V^\theta(\bx^*_{\theta',j}, j').
    \end{aligned}
    \end{equation}
    \end{small}
    As a result, we have for all $\theta' \ne \theta$
    \begin{small}
    \begin{align*}
        \sum_j &p_{\theta,j} V^\theta(\bx_{\theta,j}, j) =  \sum_j \left(1-\sqrt{\frac{\log |\Theta|}{T\delta^2}}\right) p^*_{\theta,j} V^\theta(\bx^*_{\theta,j}, j) 
         +\left(\sqrt{\frac{\log |\Theta|}{T\delta^2}}\right) p^\delta_{\theta, j}V^\theta(\bx^\delta_{\theta,j}, j) \\
        & \ge \sum_j \max_{j'} \left(1-\sqrt{\frac{\log |\Theta|}{T\delta^2}}\right) p^*_{\theta',j} V^\theta(\bx^*_{\theta',j}, j')  + \sum_j \max_{j'} \sqrt{\frac{\log |\Theta|}{T\delta^2}} p^\delta_{\theta',j} V^\theta(\bx^\delta_{\theta',j}, j') +  \delta\sqrt{\frac{\log |\Theta|}{T\delta^2}}  \\
        &\ge \sum_j \max_{j'} \Bigg(\left(1-\sqrt{\frac{\log |\Theta|}{T\delta^2}}\right) p^*_{\theta',j} V^\theta(\bx^*_{\theta',j}, j')  + \sqrt{\frac{\log |\Theta|}{T\delta^2}} p^\delta_{\theta',j} V^\theta(\bx^\delta_{\theta',j}, j')\Bigg) 
        + \delta\sqrt{\frac{\log |\Theta|}{T\delta^2}} \\
        &= \sum_j \max_{j'} p_{\theta',j} V^\theta(\bx_{\theta',j}, j') + \sqrt{\frac{\log |\Theta|}{T}}.
    \end{align*}\end{small}
    
    \noindent where the first inequality is by \eqref{eq:IC} and the second inequality is by merging two $\max$'s into one $\max$, proving the constructed mechanism $\langle \bm{p}, \bx \rangle$ is $O(\sqrt{\frac{\log |\Theta|}{T}})$-strictly IC. 

    Next, we write out the leader's expected utility for every type $\theta$
    \begin{small}
    \begin{align}\label{eq:final_obj}
        \sum_j &p_{\theta,j} U(\bx_{\theta,j}, j)  = \sum_j (1-\sqrt{\frac{\log |\Theta|}{T\delta^2}}) p^*_{\theta,j} U(\bx^*_{\theta,j}, j)  + (\sqrt{\frac{\log |\Theta|}{T\delta^2}}) p^\delta_{\theta, j}U(\bx^\delta_{\theta,j}, j) \\
        &\ge\sum_j (1-\sqrt{\frac{\log |\Theta|}{T\delta^2}}) p^*_{\theta,j} U(\bx^*_{\theta,j}, j)
    \end{align}\end{small}
    
    \noindent As a result, we have the expected utility of $\langle \bm{p}, \bx \rangle$ as $U^{\langle \bm{p}, \bx \rangle} \geq (1-\sqrt{\frac{\log |\Theta|}{T\delta^2}}) U^\RME$, proving the lemma.%
\end{proof} 

\textbf{Step 3: Existence of $(T-\bar{t})$-uniform, $O\left(\sqrt{\frac{\log |\Theta|}{T\delta^2}}\right)$ optimal, and IC  menu.} 

\begin{lemma}[\citet{althofer1994sparse}]\label{lem:approximate_menu}
    For any $\epsilon > 0$ and any $\{p_{\theta,j},\bx_{\theta, j}\}_{j\in [n]}$, there exists a $k$-uniform $\bar{\bm{p}}_\theta$ with $k=\left\lceil \frac{\log 2(|\Theta|+1)}{2\epsilon^2}\right\rceil$ such that 
    \begin{align*}
    \epsilon &\ge \left|\sum_j p_{\theta,j} U(\bx_{\theta,j}, j) - \sum_j \bar{p}_{\theta,j} U(\bx_{\theta,j}, j)\right| ,
    \end{align*}
    and for all $\theta'$
    \begin{align*}
    \epsilon &\ge \bigg|\sum_j p_{\theta,j} \max_{ j'} V^{\theta'}\big(\bx_{\theta,j}, j'\big) - \sum_j \bar{p}_{\theta,j} \max_{j'} V^{\theta'}\big(\bx_{\theta,j}, j'\big)\bigg|.
    \end{align*}
\end{lemma}

By Lemma \ref{lem:approximate_menu}, there always exists an approximation of $\bm{p}_\theta$ denoted as $\bar{\bm{p}}_\theta$ for all $\theta$, which is a $(T-\bar{t})$-uniform distribution and can be precisely fulfilled in the $T - \bar{t}$ rounds. Next, we show that the new randomized menu $\langle \bar{\bm{p}}, \bx\rangle$ is IC and $O(\sqrt{\frac{\log |\Theta|}{T\delta^2}})$ optimal. 

Since the original  $\langle \bm{p}, \bx\rangle$ is $O(\sqrt{\frac{\log |\Theta|}{T}})$-strictly IC randomized menu, we have 
    \begin{align}
        \label{eq:one}
        \sum_j p_{\theta,j} V^\theta(\bx_{\theta,j}, j)  \ge &\sum_j p_{\theta',j} \max_{j'} V^\theta(\bx_{\theta',j}, j')  + O(\sqrt{\frac{\log |\Theta|}{T}}),
    \end{align}
for all $\theta$ and $\theta' \ne \theta$.
When we approximate $\bm{p}_\theta$ with a $(T-\bar{t})$-uniform distribution $\bar{\bm{p}}_\theta$, there exists an approximate error $\epsilon = \sqrt{\frac{\log 2(|\Theta| +1)}{2(T-\bar{t})}} = O(\sqrt{\frac{\log |\Theta|}{T}})$ for all players' utilities under the randomized strategy $\{\bar{p}_{\theta,j}, \bx_{\theta, j}\}_{j \in [n]}$. As a result, for the $(T-\bar{t})$-uniform distribution $\bar{\bm{p}}_\theta$, we have  following observations by Lemma \ref{lem:approximate_menu} for all $\theta$ and $\theta' \ne \theta$: 

\begin{equation}
    \label{eq:two}
    \begin{aligned}
        \sum_j \bar{p}_{\theta,j} V^\theta(\bx_{\theta,j}, j)   
         \ge \sum_j p_{\theta,j} V^\theta(\bx_{\theta,j}, j)  - \epsilon
    \end{aligned}
\end{equation}
and
\begin{equation}
\begin{aligned}
    \label{eq:three}
    \sum_j \bar{p}_{\theta',j} \max_{j'} V^\theta(\bx_{\theta',j}, j')
     \le \sum_j p_{\theta',j} \max_{j'} V^\theta(\bx_{\theta',j}, j') + \epsilon
\end{aligned}
\end{equation}
Therefore,  we can combine \eqref{eq:one}, \eqref{eq:two}, and \eqref{eq:three}, where $\epsilon = O\left(\sqrt{\frac{\log |\Theta|}{T}}\right)$, to get 
\begin{equation}
    \begin{aligned}
        \sum_j \bar{p}_{\theta,j} V^\theta(\bx_{\theta,j}, j)  \ge \sum_j p_{\theta,j} V^\theta(\bx_{\theta,j}, j)
          - O\left(\sqrt{\frac{\log |\Theta|}{T}}\right)
    \end{aligned}
\end{equation}
and
\begin{equation}
    \begin{aligned}
        \sum_j \bar{p}_{\theta',j} \max_{j'} V^\theta(\bx_{\theta',j}, j') \le \sum_j p_{\theta,j} V^\theta(\bx_{\theta,j}, j) 
         - O\left(\sqrt{\frac{\log |\Theta|}{T}}\right)    
    \end{aligned}
\end{equation}
and equivalently, 
\begin{small}
\begin{equation}\label{eq:conclusion}
    \sum_j \bar{p}_{\theta,j} V^\theta(\bx_{\theta,j}, j)  \ge \sum_j \bar{p}_{\theta',j} \max_{j'} V^\theta(\bx_{\theta',j}, j'), \forall \theta,
\end{equation}
\end{small}
proving the new randomized menu $\langle \bar{\bm{p}}, \bx\rangle$ is IC.

Finally, the leader's utility of the new randomized menu $\langle \bar{\bm{p}}, \bx\rangle$ is also bounded by Lemma \ref{lem:approximate_menu} for all $\theta \in \Theta$ as follows
\begin{equation}
    \begin{aligned}
        \label{eq:approximately-optimal}
        \sum_j \bar{p}_{\theta,j} U(\bx_{\theta,j}, j)  \ge \sum_j p_{\theta,j} U(\bx_{\theta,j}, j)  -\epsilon =\sum_j p_{\theta,j} U(\bx_{\theta,j}, j)  - O\left(\sqrt{\frac{\log |\Theta|}{T}}\right).
    \end{aligned}
\end{equation}
Since by construction, the original randomized menu $\langle \bm{p}, \bx\rangle$ is $O\left(\sqrt{\frac{\log |\Theta|}{T\delta^2}}\right)$ optimal, we have for all $\theta \in \Theta$
\begin{small}
\begin{equation}
    \label{eq:final}
    \sum_{j} p_{\theta,j} U(\bx_{\theta,j}, j) \ge \sum_{j} p^*_{\theta,j} U(\bx^*_{\theta,j}, j) - O\left(\sqrt{\frac{\log |\Theta|}{T\delta^2}}\right).
\end{equation}
\end{small} 
Combining \eqref{eq:approximately-optimal} and \eqref{eq:final}, we also have the new randomized menu $\langle \bar{\bm{p}}, \bx\rangle$ is $O\left(\sqrt{\frac{\log |\Theta|}{T\delta^2}}\right)$ optimal, proving Step 3. Since a randomized menu with $(T-\bar{t})$-uniform distribution can be precisely simulated by a dynamic policy whereas $U^\DSE$ denotes the leader utility under the optimal dynamic policy, we have $\frac{U^{\DSE}}{T} \ge u^{\langle \bar{\bm{p}}, \bx\rangle}_l$ and proved $\frac{U^{\DSE}}{T} \geq U^{\RME} - O\left(\sqrt{\frac{\log |\Theta|}{T\delta^2}}\right)$ in the Theorem. 

\textbf{Instance with  $\Omega(1/T)$ lower bound. } We remark that the upper and lower bound of the utility comparison above is off by a factor of $O(1/\sqrt{T})$. Closing this gap is an interesting open question. To prove the $\Omega(1/T)$ bound, we consider the following game instance 
\begin{table}[tbh]
\begin{center}
\begin{small}
\begin{subtable}{0.31\linewidth}
    \centering
    \begin{tabular}{|c|c|c|} \hline
      $R$  &  $j_0$ & $j_1$  \\ \hline
      $i_0$   & 1 & 0 \\ \hline
      $i_1$   & 0 & 1 \\ \hline
    \end{tabular} 
\end{subtable}
\begin{subtable}{0.33\linewidth}
    \centering
    \begin{tabular}{|c|c|c|} \hline
      $C^0$  &  $j_0$ & $j_1$   \\ \hline
      $i_0$   & 1 & 0  \\ \hline
      $i_1$   & 1& 0 \\ \hline
    \end{tabular} 
    \end{subtable}
    \begin{subtable}{0.33\linewidth}
    \centering
    \begin{tabular}{|c|c|c|} \hline
      $C^1$  &  $j_0$ & $j_1$   \\ \hline
      $i_0$   & 0  & 1  \\ \hline
      $i_1$   & 0 & 1 \\ \hline
    \end{tabular} 
    \end{subtable}
\end{small}
\end{center}
    \label{table:counter-example}
\end{table}

\noindent in which the leader has a prior distribution $(0.5, 0.5)$ over the two follower types $(C^0, C^1)$. It is straightforward that the $\RME$ is to offer $\bx^{C^0} = (1, 0)$ with probability $1$ for type $C^0$; $\bx^{C^1} = (0, 1)$ with probability $1$ for type $C^1$. As a result, $U^{\RME} = 1$. If the leader cannot offer a menu of strategies but applies a dynamic policy instead, the optimal policy would be to use the first round to learn the follower's type based on his response (either $j_0$ or $j_1$) and play pure strategy $i_0$ or $i_1$ accordingly. The optimal expected utility at the first round, however, is at most $0.5$. In this case, we have $U^{\DSE} = 0.5 + (T-1)$ while playing $\RME$ gives the leader a utility of $T$, if the leader and the follower interact for $T$ rounds. Therefore, we get $\frac{U^{\DSE}}{T} = U^\RME - \frac{1}{2T}$, proving the theorem.%
\end{proof}

\section{Omitted Details in Section \ref{sec:algo}}
\subsection{Proof of Theorem \ref{thm:program_equivalent}}\label{append:them1-detail}


Here we will show that an optimal solution to Program \eqref{eq:U^DySS} can be recovered from an optimal solution to the following \texttt{MILP}:

{\allowdisplaybreaks
\begin{subequations}
    \label{eq:DySS_MILP}
    \begin{align}
        &\text{maximize} \quad  \sum_{\theta} \mu(\theta) \sum_{t, \bj_t, i} R_{i, j_t} \, w^{t, \theta}_{\bj_t, i} \label{eq:DySS_MILP_obj} \\
        &\text{subject to} \nonumber\\
        & z^{t, \theta}_{\bj_t} \leq y^\theta_{t, j_t}, \quad z^{t, \theta}_{\bj_t} \leq z^{t-1, \theta}_{\bj_{t-1}}, \quad z^{t, \theta}_{\bj_t} \geq y^\theta_{t, j_t} + z^{t-1, \theta}_{\bj_{t-1}} - 1, \,\,\, \forall \bj_t \in [n]^t, \, t \in [T], \, \theta \in \Theta, \label{constr:first_y}  \\
        & w^{t, \theta}_{\bj_t, i} \leq x^t_{\bj_{t-1}, i}, \quad w^{t, \theta}_{\bj_t, i} \leq z^{t, \theta}_{\bj_t}, \quad  \forall \bj_t \in [n]^t, \, t \in [T], \, \theta \in \Theta, \, i \in [m], \label{constr:third} \\
        & \sum_i x^t_{\bj_{t-1}, i} = 1, \quad \forall t \in [T], \, \bj_{t-1} \in [n]^{t-1}, \\
        & \sum_j y^\theta_{t, j} = 1, \quad \forall t \in [T], \, \theta \in \Theta, \\
        & 0 \leq a_\theta - \sum_{t,i} x^t_{\bj_{t-1}, i} C^\theta_{i, j_t} \leq M(T - \sum_{t} y_{t, j_t}^\theta), \quad \forall \bj_t \in [n]^t, \, \theta \in \Theta, \\
        & 0 \leq \bz \leq 1, \quad 0 \leq \bm{w} \leq 1, \quad 0 \leq \bx \leq 1, \quad \by \in \{0, 1\}.\label{eq:DySS_MILP_end}
    \end{align} 
\end{subequations}}

\subsubsection{A Mixed Integer Program (\texttt{MIP})}\label{sec:appendix_MIP}

Our first step is to reduce Program \eqref{eq:U^DySS} to a Mixed Integer Program (\texttt{MIP}). While this \texttt{MIP} still cannot be solved by industry-standard optimization solvers like Gurobi  \citep{gurobi}, it provides an important step towards our final development of the \texttt{MILP}.   

As mentioned above, the key barrier for practically solving Program \eqref{eq:U^DySS} is that the follower's decision variables $\bj^{t-1}_\theta$ shows up in the \emph{indices} of the leader's decision variables $\bx^t_{\bj_{t-1} } $. Our main goal in this step is to remove this dependence.  To do so, we introduce a different variable representation for the follower's action space. Specifically, we use a binary matrix $ \by^\theta \in \{0, 1\}^{T \times n} $ to represent the  response sequence for any follower type $\theta \in \Theta$.   

Moreover,  the $t$'th row $\by^\theta_t \in \{0, 1\}^n$ is a \textit{one-hot} vector of length $n$, in which the index of the $1$ entry represents the action taken by the follower at time step $t$.
The key step of the reduction is to characterize the follower's optimal response history by the following constraint\footnote{We omit the feasible regions of variables (e.g. $t \in [T], i \in [m], \bj_T \in [n]^T, \theta \in \Theta$) in some programs for the sake of space when they are clear from the context.}
\begin{equation}\label{eq:optimal_follower_response_main}
    0 \leq a_\theta - \sum_{t, i} x^t_{\bj_{t-1},i} C^\theta_{i,j_t}  \leq M(T - \sum_{t=1}^T y_{t, j_t}^\theta), \, \forall \bvec{j}_T, \theta \,
\end{equation}
\noindent where $M$ is a very large constant and $\bm{a} \in \mathbb{R}^k$ is a set of newly introduced \emph{decision variables}.  Formally,  Program \eqref{eq:U^DySS} can be transformed into a $\MIP$ with an objective of $\sum_{\theta} \mu(\theta) \sum_{t, \bj_{t}, i} R_{i, j_t} \, x^t_{\bj_{t-1},i} \prod_{t'=1}^{t} y^\theta_{{t'}, j_{t'}}$. 

\begin{lemma}\label{lem:optimal_follower_response}
Program \eqref{eq:U^DySS} is equivalent to the following $\MIP$ \eqref{eq:Dyss_MIP}. 
\begin{subequations}
    \label{eq:Dyss_MIP}
    \begin{align}
        &\text{maximize} \quad  \sum_{\theta} \mu(\theta) \sum_{t, \bj_{t}, i} R_{i, j_t} \, x^t_{\bj_{t-1},i} \prod_{t'=1}^{t} y^\theta_{{t'}, j_{t'}} \label{eq:Dyss_MIP_obj} \\
        &\text{subject to} \quad  \nonumber \\
        & 0 \leq a_\theta - \sum_{t,i} x^t_{\bj_{t-1},i} C^\theta_{i,j_t} \leq M\left(T - \sum_{t \in [T]} y_{t, j_t}^\theta \right), \quad \forall \bj_T \in [n]^T, \, \theta \in \Theta, \\
        & \sum_i x^t_{\bj_{t-1},i} = 1, \quad \forall \bj_{t-1} \in [n]^{t-1}, \, t \in [T], \\
        & x^t_{\bj_{t-1},i} \in [0, 1], \quad \forall i \in [m], \, \bj_{t-1} \in [n]^{t-1}, \, t \in [T], \\
        & \sum_j y^\theta_{t, j} = 1, \quad \forall t \in [T], \, \theta \in \Theta, \\
        & y^\theta_{t, j} \in \{0, 1\}, \quad \forall j \in [n], \, t \in [T], \, \theta \in \Theta.\label{eq:Dyss_MIP_end}
    \end{align}
\end{subequations}
where $\bx, \by, \ba$ are decision variables and  $M$ is a  sufficiently large  constant.
\end{lemma}

\begin{proof}{Proof of Lemma~\ref{lem:optimal_follower_response}.}
\textbf{Reformulating follower's incentive constraints in Program \eqref{eq:U^DySS}. } We argue that under our new variable $ \by^\theta \in \{0, 1\}^{T \times n} $,  the following constraint is equivalent to the first set of constraints in Program \eqref{eq:U^DySS} and thus correctly capture the follower's dynamic optimal responses:
\begin{small}
\begin{align}\label{eq:optimal_follower_response}
    0 \leq a_\theta - \sum_{t, i} x^t_{\bj_{t-1},i} C^\theta_{i,j_t}  \leq M(T - \sum_{t=1}^T y_{t, j_t}^\theta), \, \forall \bvec{j}_T, \theta, 
\end{align}
\end{small}
where $M$ is a very large constant, $\bm{a} \in \mathbb{R}^k$ is a set of newly introduced \emph{decision variables}.  

To see how Inequality \eqref{eq:optimal_follower_response} characterizes the optimality condition of the follower's best responses, consider any  variables $\bx, \by$ and $\bm{a}$ feasible to Inequality \eqref{eq:optimal_follower_response}. Recall that $\by^\theta_t \in \{0, 1\}^n$ is a one-hot vector of length $n$ for any $\theta \in \Theta$ and $t \in [T]$. We denote $j^{*\theta}_{t} \in [n]$ as the  index of the unique $1$-valued entry in $\by^\theta_t$, and vector $\bj_{T}^{*\theta} = (j^{*\theta}_1,\cdots,j^{*\theta}_{T})$. We argue that if $\bx, \by$ and $\bm{a}$ satisfy Inequality \eqref{eq:optimal_follower_response}, then $\bj_{T}^{*\theta}$ must be an optimal response sequence for type $\theta$. 

By construction, we have $\sum_{t}y_{t, j^{*\theta}_{t}}^\theta = T,$ for all $\theta$ and $\bj^{*\theta}_T$ by  definition. Plugging this  into Inequality \eqref{eq:optimal_follower_response}, we obtain   $M(T - \sum_{t}y_{t, j^{*\theta}_{t}}^\theta) = 0$. Thus  Inequality \eqref{eq:optimal_follower_response} becomes $ 
0  \leq (a_\theta - \sum_{t\in[T]} \sum_{i \in [m]} x^{t}_{\bj^{*\theta}_{t-1},i} C^\theta_{i,j^{*\theta}_{t}})\leq 0$, 
which implies 
\begin{equation}\label{eq:best_response}
a_\theta = \sum_{t\in[T]} \sum_{i \in [m]} x^t_{\bj^{*\theta}_{t-1},i} C^\theta_{i,j^{*\theta}_{t}}, \quad \forall \theta.
\end{equation}
That is, $a_{\theta}$ is precisely follower type $\theta$'s total utility from response $\by^\theta$.

Now let us consider any other possible response sequence $\bj_T \in [n]^T \neq \bj_{T}^{*\theta}$. Since $\by^\theta_t$ is a one-hot vector and $y_{t, j_t}^\theta = 1$ if and only if $j_t = j^{*\theta}_{t}$. Therefore, if $\bj_T \neq \bj^{*\theta}_T$, there must exists some $t$ such that $y_{t, j_t}^\theta = 0$, hence, \[\sum_{t}y_{t, j_t}^\theta < T, \quad \forall \theta, \bj_T \neq \bj_{T}^{*\theta}.\] Plugging this $\bj_T$ into \eqref{eq:optimal_follower_response}, we must have $M(T - \sum_{t}y_{t, j_t}^\theta) \geq M$ as a very large constant, which makes the right-hand-side inequality void for any $\bj_T \in [n]^T \neq \bj_{T}^{*\theta}$.  Thus, the only useful constraint is $ 0 \leq (a_\theta - \sum_{t\in[T]} \sum_{i \in [m]} x^t_{\bj_{t-1},i} C^\theta_{i,j_t})$ for any $ \bj_T \neq \bj_{T}^{*\theta}$,
which implies
\begin{equation}\label{eq:non_optimal_response}
\sum_{t\in[T]} \sum_{i \in [m]} x^t_{\bj_{t-1},i} C^\theta_{i,j_t} \leq a_\theta, \quad \forall \theta, \bj_T \neq \bj_{T}^{*\theta}.
\end{equation}

It is now clear to see that $\bj_{T}^{*\theta}$ is the optimal response for type $\theta$ by combining \eqref{eq:non_optimal_response} and \eqref{eq:best_response}. 

\textbf{Reformulating  the objective function in Program \eqref{eq:U^DySS}.}  We claim that for any feasible variable    $\bx, \by$ and $\bm{a}$, the objective of Program   \eqref{eq:U^DySS} equals to the following expression: 

\begin{equation}\label{eq:obj-reformulate}
    \sum_{\theta} \mu(\theta) \sum_{t \in [T], i \in [m], \bj_{t} \in [n]^t} R_{i, j_t} \,x^t_{\bj_{t-1},i} \prod_{t'=1}^{t} y^\theta_{{t'}, j_{t'}}
\end{equation}
The objective above enumerates all possible path histories $\bj_{t} \in [n]^t$ for any $t \in [T]$. However, the product $\prod_{t'=1}^{t} y^\theta_{{t'}, j_{t'}}$ is only non-zero when all $\{ y^\theta_{{t'}, j_{t'}} \}_{t' = 1}^t $ are non-zero, i.e. when $\{ j_{t'} \}_{t'=1}^t$ is on the optimal follower response path for follower type $\theta$. This guarantees that Objective \eqref{eq:obj-reformulate} only counts the $\bj_{t}$ on the optimal path $\bj_{T}^{*\theta}$. Therefore, it correctly calculates the expected leader utility when each follower type $\theta$ follows their optimal response paths.

\textbf{Showing the equivalence.} Consider $\bx$ and $\bj$ as a feasible solution of \eqref{eq:U^DySS}. We will show that 
\begin{small}
    \begin{align*}
    \bx, \quad y^\theta_{t,j} = \begin{cases}1 \text{ if } j = j^{\theta}_t \\ 0 \text{ otherwise,}\end{cases} \text{and } a_\theta = \sum_{t\in[T]} u^f_\theta(\bx^t_{\bj_{t-1}^\theta}, j^\theta_t)
    \end{align*} 
\end{small}

\noindent is a feasible solution of \eqref{eq:Dyss_MIP} of the same objective function value. The last four constraints of \eqref{eq:Dyss_MIP} are satisfied by construction. To see 
how the first constraint is satisfied, note for any $\bj_T \ne \bj_T^\theta$ returned by \eqref{eq:U^DySS} of type $\theta$, we have $\sum_{t\in[T]}y_{t, j_t}^\theta < T$ by construction, and its corresponding follower utility $\sum_{t,i} x^t_{\bj_{t-1},i} C^\theta_{i,j_t} \leq a_\theta$ by the definition that $\bj_T$ is not the best response. When $\bj_T = \bj_T^\theta$, we have  $M(T - \sum_{t\in[T]}y_{t, j_t}^\theta) = 0$ and $ a_\theta = \sum_{t,i} x^t_{\bj_{t-1},i} C^\theta_{i,j_t}$ by construction. Hence, the first constraint of \eqref{eq:Dyss_MIP} is also satisfied. The fact that $y^\theta_{t, j} = 0$ guarantees that $\prod_{t'=1}^{T} y^\theta_{{t'}, j_{t'}} = 0$ if $\bj_T \ne \bj^\theta_T$. Then we have $ \sum_{t, \bj_{t}, i} R_{i, j}\big( x^t_{\bj_{t-1},i} \prod_{t'=1}^{t} y^\theta_{{t'}, j_{t'}} \big) = \sum_{t, i} R_{i, j^\theta_t} x^t_{\bj^\theta_{t-1},i} = \sum_t U(\bx^t_{\bj_{t-1}^\theta}, j^\theta_t)$, which shows the constructed feasible solution of \eqref{eq:Dyss_MIP} has the same objective value as \eqref{eq:U^DySS}.

Let us now consider $\bx$, $\by$, and $\bm{a}$ feasible for \eqref{eq:Dyss_MIP}. We construct $\bj^\theta_T$ such that $j^\theta_t \in [n]$ and $y^\theta_{t, j^\theta_t} = 1$. We will show that $\bx$ and $\by$  are feasible for \eqref{eq:U^DySS} with the same objective value. Recall the discussion from equations \eqref{eq:best_response}-- \eqref{eq:non_optimal_response}, we have that $\bj^\theta_T$ captures the follower's optimal behavior and satisfies $\sum_{t\in[T]} V^\theta(\bx^t_{\bj_{t-1}^\theta}, j^\theta_t) \geq \sum_{t\in[T]} V^\theta(\bx^t_{\hat{\bj}_{t-1}}, \hat{j}_t)$. What is more, by the same argument as in the previous direction, we have $  \sum_{t, \bj_{t}, i} R_{i, j}\big( x^t_{\bj_{t-1},i} \prod_{t'=1}^{t} y^\theta_{{t'}, j_{t'}} \big) = \sum_t U(\bx^t_{\bj_{t-1}^\theta}, j^\theta_t)$, which shows the equivalence  between the objective values of these two programs.%
\end{proof}

\subsubsection{Connection Between the \texttt{MIP} \eqref{eq:Dyss_MIP} and \texttt{MILP} \eqref{eq:DySS_MILP}}
\label{append:MIP_and_MILP}

Our second step is to show the connection between the previous \texttt{MIP} \eqref{eq:Dyss_MIP} and our \texttt{MILP} \eqref{eq:DySS_MILP}. First of all, we present how to linearize the product term $  x^t_{\bj_{t-1},i} \prod_{t'=1}^{t} y^\theta_{{t'}, j_{t'}}$, we introduce two sets of auxiliary  decision variables $\bm{z}$ and $\bm{w}$.  Specifically, 
\begin{equation}
    \label{eq:variable-def}
    \begin{aligned}
        z^{t,\theta}_{\bj_t} &=  \prod_{t'=1}^{t} y^\theta_{{t'}, j_{t'}}, \forall \theta \in \Theta, t \in [T],  \bj_t \in [n]^t, \\
        w^{t,\theta}_{\bj_t,i} &=x^t_{\bj_{t-1}, i} \cdot \prod_{t'=1}^{t} y^\theta_{{t'},  j_{t'}}, \forall \theta , i  , t  , \bj_t \in [n]^t. 
    \end{aligned}
\end{equation}

The key challenge for deriving our re-formulation is to set up the right set of constraints for these new variables $\bz, \bm{w}$ so that they will exactly  enforce the feasibility of the original variables. Constraints \eqref{constr:first_y} -- \eqref{constr:third} in Program \eqref{eq:DySS_MILP} serve the purpose, which leads to our following main theorem and, consequently, a practical formulation for computing the  \DSE. 
Note that $\MILP$ \eqref{eq:DySS_MILP} and $\MIP$  \eqref{eq:Dyss_MIP} themselves are not, in general, equivalent since not all feasible solutions of one correspond to feasible solutions of the other. 
However, we show that we can convert any \emph{optimal} solution of one program into a corresponding feasible \emph{optimal} solution for the other. The proof is as follows.

\begin{proof}{Proof of the connection between $\MILP$ \eqref{eq:DySS_MILP} and $\MIP$  \eqref{eq:Dyss_MIP}.} 
The proof has two steps. We first show that any optimal solution for  $\MIP$  \eqref{eq:Dyss_MIP} must correspond to a feasible solution of $\MILP$ \eqref{eq:DySS_MILP} with the same objective value, and then show its reverse direction. 

\textbf{ Step 1: $OPT\eqref{eq:DySS_MILP} \geq OPT\eqref{eq:Dyss_MIP}$. }  Consider any optimal solution $\bx$, $\by$, and $\bm{a}$ for  $\MIP$  \eqref{eq:Dyss_MIP}. We show that $\bx$, $\by$, $\bm{a}$, together with constructed $\bz, \bm{w}$ via  Equations \eqref{eq:variable-def}, forms a feasible solution of $\MILP$ \eqref{eq:DySS_MILP}. This follows from relatively standard algebraic derivations.  Consider any optimal solution $\bx$, $\by$, and $\bm{a}$ for  MIP \eqref{eq:Dyss_MIP}. We argue that $\bx$, $\by$, $\bm{a}$, together with constructed $\bz, \bm{w}$ via  Equations \eqref{eq:variable-def}, forms a feasible solution of \eqref{eq:DySS_MILP}.  The equivalence of the objective function is immediately satisfied by construction. It is also obvious that $0\leq \bz \leq 1 \text{ and } 0\leq \bm{w} \leq 1$ by construction. Next, we show that constraints \eqref{constr:first_y} -- \eqref{constr:third} are satisfied. Recall that each entry of $\by$ is   binary. For any $\theta$, $t$ and $(j_1,\cdots, j_t)$, we consider the following two possible cases:
\begin{itemize}
    \item If there exists some $j_i \in (j_1,\cdots, j_t)$ such that $y^\theta_{i, j_i} = 0$, then we have $z^{t,\theta}_{\bj_i} = 0$ by construction. Then, we see by induction of $z^{t',\theta}_{\bj_{t'}}\leq z^{t',\theta}_{\bj_{t'-1}}$ for $t' \in [i + 1, \cdots, t]$ that $z^{t,\theta}_{\bj_t} = 0$ and consequently $w^{t,\theta}_{\bj_t,i} = 0$ must hold as well. As a result, it is obvious that the first two constraints of \eqref{constr:first_y} and \eqref{constr:third} are satisfied. To see why the last constraint of \eqref{constr:first_y} holds, note that as long as at least one  of the $j_i \in (j_1,\cdots, j_t)$ satisfy $y^\theta_{i, j_i} = 0$, then at least one of $y^\theta_{t, j_t}$ and $z^{t,\theta}_{\bj_{t-1}}$ must be $0$ as well, meaning $y^\theta_{i, j_t} + z^{t,\theta}_{\bj_{t-1}} \le 0 = z^{t,\theta}_{\bj_t}$.
    \item If $y^\theta_{i, j_i} = 1$ for all $i=1,\cdots,t$, then we have $z^{t,\theta}_{\bj_t} = 1$ and $w^{t,\theta}_{\bj_t,i} = x^t_{\bj_{t-1}, i}$ by construction. It is also straightforward to see that constraint the first two constraints of \eqref{constr:first_y} and \eqref{constr:third} are satisfied. Moreover, we have $\sum_{i\in [t]} y^\theta_{i, j_i} = t$ in this case, and thus $\sum_{i\in [t]} y^\theta_{i, j_i} - (t-1) = 1 = z^{t,\theta}_{\bj_t}$ so  the last constraint of \eqref{constr:first_y} is satisfied as well.
\end{itemize}
Therefore, we have shown the constructed $\bz,\bm{w}$ together with  $\bx, \by, \ba$ form a feasible solution to  MILP \eqref{eq:DySS_MILP} with the same optimal objective value as \eqref{eq:Dyss_MIP}. This implies $OPT\eqref{eq:DySS_MILP} \geq OPT\eqref{eq:Dyss_MIP}$, where $OPT$ denotes the optimal solution value of the program.

\textbf{ Step 2: $OPT\eqref{eq:DySS_MILP} \leq OPT\eqref{eq:Dyss_MIP}$. }  We now prove the reverse direction, which is the more interesting and non-trivial step.  Specifically, suppose that we are given an optimal solution $\bx, \by, \bm{a}, \bz, \bm{w}$ of $\MILP$ \eqref{eq:DySS_MILP}. We show that these $\bx, \by$ and $\bm{a}$ form a feasible solution of \eqref{eq:Dyss_MIP} with the same objective function. It is obvious to see that $\bx, \by$ and $\bm{a}$ are still feasible to \eqref{eq:Dyss_MIP} since the constraints of \eqref{eq:Dyss_MIP} is a subset of constraints of \eqref{eq:DySS_MILP}. The key step is to show these $\bx, \by$ and $\bm{a}$ achieve the same objective value in \eqref{eq:Dyss_MIP}. To show this, we prove that for any optimal solution $\bx, \by, \bm{a}, \bz, \bm{w}$ of \eqref{eq:DySS_MILP}, we must have $w^{t,\theta}_{\bj_t,i} = x^t_{\bj_{t-1}, i} \cdot \prod_{t'=1}^{t} y^\theta_{{t'}, j_{t'}}$ for any $\theta, t, \text{ and } \bj_t$. We prove this is correct through a careful case analysis:

(1) If there exists some $j_i \in (j_1,\cdots, j_t)$ such that $y^\theta_{i, j_i} = 0$, then $w^{t,\theta}_{\bj_t,i} = x^t_{\bj_{t-1}, i} \cdot \prod_{t'=1}^{t} y^\theta_{{t'}, j_{t'}}=0$.

(2) If $y^\theta_{i, j_i} = 1$ for all $i=1,\cdots,t$, we claim  $w^{t,\theta}_{\bj_t,i} = x^t_{\bj_{t-1}, i} \cdot \prod_{t'=1}^{t} y^\theta_{{t'}, j_{t'}} = x^t_{\bj_{t-1}, i}$. For the sake of contradiction, let us assume    $w^{t,\theta}_{\bj_t,i} \neq  x^t_{\bj_{t-1}, i}$.   Note that we must have  $w^{t,\theta}_{\bj_t,i} \leq  x^t_{\bj_{t-1}, i}$ and $w^{t,\theta}_{\bj_t,i} \le z^{t,\theta}_{\bj_t}$  by constraint \eqref{constr:third}. Note that when $y^\theta_{i, j_i} = 1$ for $i=1,\cdots,t$,   we have $z^{t,\theta}_{\bj_t} \leq 1$ by the first two constraints of \eqref{constr:first_y}, as well as $z^{t,\theta}_{\bj_t} \geq 1$ by the last constraint of \eqref{constr:first_y}. Hence, $z^{t,\theta}_{\bj_t} = 1$,  and the second constraint of \eqref{constr:third} will not add additional constraint on $w^{t,\theta}_{\bj_t,i}$. By assumption $w^{t,\theta}_{\bj_t,i} \neq  x^t_{\bj_{t-1}, i}$, then we must have $w^{t,\theta}_{\bj_t,i} <  x^t_{\bj_{t-1}, i}$ because of the first constraint of \eqref{constr:third}.   In this case, we can always construct another $\bm{\hat{w}}$ such that $\hat{w}^{t,\theta}_{\bj_t,i} = x^t_{\bj_{t-1}, i}$. The new $\hat{\bm{w}}$ satisfies all the constraints, so it would still be feasible but has a higher objective value \eqref{eq:DySS_MILP} since $R_{i,j} \in [0, 1]$, which contradicts with $\bm{w}$ is part of an optimal solution.
As a result, we have shown that for any optimal solution $\bx, \by, \bm{a}, \bz, \bm{w}$ of $\MILP$ \eqref{eq:DySS_MILP}, we must have $w^{t,\theta}_{\bj_t,i} = x^t_{\bj_{t-1}, i} \cdot \prod_{t'=1}^{t} y^\theta_{{t'}, j_{t'}}$. Thus, $\bx, \by, \bm{a}$ form a feasible solution to $\MIP$  \eqref{eq:Dyss_MIP} with the same objective value, implying $OPT\eqref{eq:DySS_MILP} \leq OPT\eqref{eq:Dyss_MIP}$.%
\end{proof}

\subsection{\texttt{MILP}s to Compute the \texttt{DSE} Approximately}\label{sec:approximate_programs}
\subsubsection{$\MILP$ to compute the $\Markov$ policy.}
\begin{subequations}
    \label{eq:U^Markov}
    \begin{align}
    &\text{maximize} \quad  \sum_{t \in [T]} \sum_{\theta \in \Theta} \Big[ \mu(\theta) \, U(\bx^t_{j_{t-1}^\theta}, j^\theta_t) \Big] \\
    &\text{subject to}  \\
    & \sum_{t \in [T]} V^\theta(\bx^t_{j_{t-1}^\theta}, j^\theta_t) 
    \geq \sum_{t \in [T]} V^\theta(\bx^t_{\hat{j}_{t-1}}, \hat{j}_t), \quad \forall \theta \in \Theta, \, \hat{\bj}_T \in [n]^T, \bj_T^\theta \in [n]^T, \forall \theta \in \Theta. 
\end{align}
\end{subequations}

\subsubsection{$\MILP$ to compute the \textsc{First}-k policy.}
\begin{subequations}
    \label{eq:U^fixed1}
    \begin{align}
        &\text{maximize} \quad  \sum_{\theta \in \Theta} \mu(\theta) \Bigg[ \sum_{t \in [k]} U(\bx^t_{\j_{t-1}^\theta}, j^\theta_t) + U(\bx^{k+1}_{\j_{k}^\theta}, j^\theta_{k+1}) \cdot (T-k) \Bigg] \\
        &\text{subject to} \quad \\
        & \sum_{t \in [k]} V^\theta(\bx^t_{\j_{t-1}^\theta}, j^\theta_t) + V^\theta(\bx^{k+1}_{\j_{k}^\theta}, j^\theta_{k+1}) \cdot (T-k)  \ge \sum_{t \in [k]} V^\theta(\bx^t_{\hat{j}_{t-1}}, \hat{j}_t) + V^\theta(\bx^{k+1}_{\hat{\j}_{k}}, \hat{j}_{k+1}) \cdot (T-k), \notag 
        \\ &\qquad \forall \theta \in \Theta, \, \hat{\bj}_{k+1} \in [n]^{k+1}, \bj_{k+1}^\theta \in [n]^{k+1}, \forall \theta \in \Theta.
    \end{align}
\end{subequations}

\section{Additional Experimental Results}\label{appendix_additionalExp}
In this section, we provide some additional experimental results not included in Section~\ref{sec:experiment}.

\subsection{Additional Results on Structured Games}

\paragraph{The Dynamic Pricing Game.} We compute the exact $\DSE$ using the MILP from Section \ref{sec:algo} on a dynamic pricing game, as described in Section~\ref{section:dynamic_pricing_games}, to experimentally verify the no learning theorem.
We examine a game with finite buyer value set $V = \{0.4, 0.5, 0.6\}$ and a uniform value distribution. 
It is easy to compute that the optimal single-round mechanism is to post the Myerson price of $0.4$, inducing expected revenue of $0.4$. 
This game can be written as a Stackelberg game, with payoff matrices given by 
Figure~\ref{table:dynamic_pricing_experiment_game}.
Here, $R$ represents the seller's utility matrix; $C^0, C^1, C^2$ represent the buyer's utility matrix when $v=0.4, v=0.5, v=0.6$ accordingly; $i_0, i_1, i_2$ represent the seller sets a posted price $0.4, 0.5, 0.6$; $j_0, j_1$ corresponds to the buyer rejecting or purchasing the item. 
We run experiments for $T = 1,..., 10$, and observe  $U^\DSE = 0.4 T$, as expected. That is, the average optimal dynamic pricing utility is indeed the same as the optimal static revenue. 


\begin{figure}[tb]
\begin{small}
\begin{subtable}{0.24\linewidth}
    \centering
    \begin{tabular}{|c|c|c|} \hline
      $R$  &  $j_0$ & $j_1$  \\ \hline
      $i_0$  & 0 & 0.4 \\ \hline
      $i_1$  & 0 & 0.5 \\ \hline
      $i_2$  & 0 & 0.6 \\ \hline
    \end{tabular} 
\end{subtable}
\begin{subtable}{0.24\linewidth}
    \centering
    \begin{tabular}{|c|c|c|} \hline
      $C^0$  &  $j_0$ & $j_1$  \\ \hline
      $i_0$  & 0 & 0 \\ \hline
      $i_1$  & 0 & -0.1 \\ \hline
      $i_2$  & 0 & -0.2 \\ \hline
    \end{tabular} 
    \end{subtable}
\begin{subtable}{0.24\linewidth}
    \centering
    \begin{tabular}{|c|c|c|} \hline
      $C^1$  &  $j_0$ & $j_1$  \\ \hline
      $i_0$   & 0& 0.1 \\ \hline
      $i_1$  & 0 & 0 \\ \hline
      $i_2$  & 0 & -0.1 \\ \hline
    \end{tabular} 
\end{subtable}
\begin{subtable}{0.24\linewidth}
    \centering
    \begin{tabular}{|c|c|c|} \hline
      $C^2$  &  $j_0$ & $j_1$  \\ \hline
      $i_0$  & 0 & 0.2 \\ \hline
      $i_1$  & 0 & 0.1 \\ \hline
      $i_2$  & 0 & 0 \\ \hline
    \end{tabular} 
    \end{subtable}
    \end{small}
\caption{Utility Matrices for a Dynamic Pricing Game. A pricing game with three buyer types, $V = \{0.4, 0.5, 0.6\}$, from a uniform distribution. The sellers actions are to set a price of $\{0.4, 0.5, 0.6\}$.}
\label{table:dynamic_pricing_experiment_game} 
\end{figure}

\paragraph{Stackelberg security game.}
Additionally, we consider the case where the leader has one more target that needs to be protected (i.e. more actions for both the leader and the follower). Specifically, consider the following utility matrices for both agents.

\begin{table}[tbh]
\begin{small}
\begin{minipage}{0.24\linewidth}
    \centering
    \begin{tabular}{|c|c|c|c|} \hline
      $R$  &  $t_0$ & $t_1$  & $t_2$ \\ \hline
      $t_0$   & 1 & 0& 0\\ \hline
      $t_1$   & 0& 1 & 0\\ \hline
      $t_2$   & 0& 0& 1 \\ \hline
    \end{tabular} 
\end{minipage}
\begin{minipage}{0.24\linewidth}
    \centering
    \begin{tabular}{|c|c|c|c|} \hline
      $C^0$  &  $t_0$ & $t_1$ & $t_2$ \\ \hline
      $t_0$   & 0 & 0.5 & 0.5 \\ \hline
      $t_1$   & 1 & 0 & 0.5 \\ \hline
      $t_2$   & 1 & 0.5 & 0 \\ \hline
    \end{tabular} 
    \end{minipage}
\begin{minipage}{0.24\linewidth}
    \centering
    \begin{tabular}{|c|c|c|c|} \hline
      $C^1$  &  $t_0$ & $t_1$  & $t_2$ \\ \hline
      $t_0$   & 0 & 1 & 0.5\\ \hline
      $t_1$   & 0.5 & 0 & 0.5 \\ \hline
      $t_2$   & 0.5 & 1 & 0 \\ \hline
    \end{tabular} 
\end{minipage}
\begin{minipage}{0.24\linewidth}
    \centering
    \begin{tabular}{|c|c|c|c|} \hline
      $C^2$  &  $t_0$ & $t_1$ & $t_2$ \\ \hline
      $t_0$   & 0& 0.5 & 1 \\ \hline
      $t_1$   & 0.5 & 0& 1 \\ \hline
      $t_2$   & 0.5 & 0.5 & 0\\ \hline
    \end{tabular} 
    \end{minipage}
    \end{small}
    \caption{Utility Matrices for a Stackelberg Security Game with three follower types.}
\end{table}

The leader has 1 unit of resource to protect three targets from the attacker whose type is drawn uniformly from $\{C^0, C^1, C^2\}$. Each follower type $C^i, i\in \{0, 1, 2\}$ prefers target $t_i$.  Table \ref{table:ssg_appendix} shows the leader's expected average utility in the dynamic setup for different numbers of interaction rounds. Note the leader's expected average utility dropped overall compared to the result in the main paper, which is reasonable since the leader has to use the same resource to protect more targets.

\begin{table}[tbh]
\centering
\begin{small}
\begin{tabular}{|c|c|c|} 
 \hline
   & $U^\DSE/T$  & $U^\RME$ \\ 
 \hline
 T = 1  & 1/3  & 1/3\\
 \hline
 T = 2  & \textbf{0.444}  & 1/3\\
 \hline
 T = 3  & \textbf{0.467}   & 1/3\\
 \hline
 T = 4  & \textbf{0.479}  & 1/3 \\
 \hline
 T = 5  & \textbf{0.493}  & 1/3\\
 \hline
\end{tabular}
\end{small}
\caption{Average Leader Utility Per Round. The average utility per round for both the $\DSE$ and $\RME$ is shown for $T = 1, ..., 5$.}
\label{table:ssg_appendix}
\end{table}

\subsection{Additional Results on Structured Games}
In Table \ref{table:m3n3k3_table}, we provide further experimental results on smaller random game instances compared to those discussed in the main paper. We continue to compare the runtime and average leader utility for solving the optimal $\Markov$ policy and the \textsc{First}-k policy, against solving the $\DSE$. These additional experiments lead to the same conclusion as in the main paper, reinforcing our previous findings.

\begin{table}[tbh]
\begin{small}
\begin{tabular}{|c|c|c|c||c|c|c|c|} 
     \hline
     & \multicolumn{3}{c||}{Runtime} & \multicolumn{4}{c|}{Average Utility} \\
     \hline
     & $\First$& $\Markov$ & $\DSE$ & $\First$ & $\Markov$  & $\DSE$ & $\RME$ \\
     \hline
     T = 1 & $0.02 \pm .002$ & $0.02 \pm .004$ & $0.02 \pm .004$ & $0.66 \pm .12$ & $0.66 \pm .12$& $0.66 \pm .12$ & $\mathbf{0.72} \pm .12$\\
     \hline
     T = 2 & $\mathbf{0.44}\pm 0.3$ & $0.48\pm 0.4$ & $\mathbf{0.44} \pm 0.28$ & $\mathbf{0.73} \pm .11$ & $\mathbf{0.73} \pm .11$& $\mathbf{0.73} \pm .11$ &$0.72 \pm .12$ \\
     \hline
     T = 3 & $6.5 \pm 1.9$ & $\mathbf{2.4} \pm 1.4$ & $6.6 \pm 1.9$ & $\mathbf{0.76} \pm .11$ & $0.75 \pm .11$ & $\mathbf{0.76} \pm .11$   & $0.72 \pm .12$ \\
     \hline
     T = 4 &$79 \pm 49$ & $\mathbf{11} \pm 11$  & $79 \pm 49$ & $\mathbf{0.78} \pm .11$  &$0.76 \pm .11$  &$\mathbf{0.78} \pm .11$   &$0.72 \pm .12$  \\
     \hline
     T = 5 & $104 \pm 48$ & $\mathbf{75} \pm 85$  & $1247 \pm 1031$ & $0.78 \pm .11$ & $0.76 \pm .11$ & $\mathbf{0.79} \pm .11$ & $0.72 \pm .12$ \\
     \hline
     T = 6 & $\mathbf{125}$  $\pm 46$  & $658$  $\pm  1060$ & N/A & $\mathbf{0.78} \pm .11$ & $0.76 \pm .12$  & N/A & $0.72 \pm .12$  \\
     \hline
     T = 7 & $ \mathbf{141}$  $\pm 65$  & $ 6657  \pm 1213$  & N/A & $\mathbf{0.78} \pm .11$ & $0.77 \pm .11$ & N/A & $0.72 \pm .12$ \\
     \hline
     T = 8 & $ \mathbf{160}$  $\pm 59$  & N/A & N/A & $\mathbf{0.78} \pm .11$ & N/A  & N/A & $0.72 \pm .12$ \\
     \hline
\end{tabular}
\end{small}
\caption{Additional experimental results on randomized games.  Running time (columns 2-4 with the unit: second) and average utility (columns 5-8) for random game instances with $m=3, n=3, |\Theta|=3$, $k=3$, where ``N/A'' means the algorithm can not return a solution within 3 hours.} 
\label{table:m3n3k3_table}
\end{table}

\end{document}